\documentclass{IEEEtran}
\usepackage{subfigure}
\usepackage{float}
\usepackage{lipsum}
\usepackage{leftidx}
\usepackage[utf8]{inputenc}
\usepackage{color}
\usepackage{graphicx}
\usepackage{amssymb}
\usepackage{multicol}
\usepackage{amsmath}
\usepackage{amsthm}

\usepackage{mathtools}

\newtheorem{cor}{\textit{Corollary}}

\newtheorem{theo}{\textit{Theorem}}

\newcommand\blfootnote[1]{%
  \begingroup
  \renewcommand\thefootnote{}\footnote{#1}%
  \addtocounter{footnote}{-1}%
  \endgroup
}

\begin{document}
\title{Spectral Efficiency of Mixed-ADC Massive MIMO
}
\author{Hessam Pirzadeh, \emph{Student Member}, \emph{IEEE}, and A. Lee Swindlehurst, \emph{Fellow}, \emph{IEEE}}

\maketitle
\begin{abstract}
We study the spectral efficiency (SE) of a mixed-ADC massive MIMO system in which $K$ single-antenna users communicate with a base station (BS) equipped with $M$ antennas connected to $N$ high-resolution ADCs and $M-N$ one-bit ADCs. This architecture has been proposed as an approach for realizing massive MIMO systems with reasonable power consumption. First, we investigate the effectiveness of mixed-ADC architectures in overcoming the channel estimation error caused by coarse quantization. For the channel estimation phase, we study to what extent one can combat the SE loss by exploiting just $N\ll M$ pairs of high-resolution ADCs. We extend the round-robin training scheme for mixed-ADC systems to include both high-resolution and one-bit quantized observations. Then, we analyze the impact of the resulting channel estimation error in the data detection phase. We consider random high-resolution ADC assignment and also analyze a simple antenna selection scheme to increase the SE. Analytical expressions are derived for the SE for maximum ratio combining (MRC) and numerical results are presented for zero-forcing (ZF) detection. Performance comparisons are made against systems with uniform ADC resolution and against mixed-ADC systems without round-robin training to illustrate under what conditions each approach provides the greatest benefit.
%\vspace{-1mm}
\end{abstract}

\blfootnote{This work was supported by the National Science Foundation under Grants ECCS-1547155 and CCF-1703635, and by a Hans Fischer Senior Fellowship from the Technische Universit\"at
M\"unchen Institute for Advanced Study.

H. Pirzadeh and A. L. Swindlehurst are with the Center for Pervasive
Communications and Computing, University of California, Irvine, CA 92697
USA (e-mail: hpirzade@uci.edu; swindle@uci.edu).

Portions of this paper have appeared in \cite{Hessam1}.}

%\vspace{-7mm}

\section{Introduction}\label{sec:Introduction}

%\vspace{-1mm}

\IEEEPARstart{T}{he} seminal work of Marzetta introduced massive MIMO as a promising architecture for future wireless systems \cite{Marzetta}. In the limit of an infinite number of base station (BS) antennas, it was shown that massive MIMO can substantially increase the network capacity. Another key potential of massive MIMO systems which has also made it interesting from a practical standpoint is its ability of achieving this goal with inexpensive, low-power components \cite{LuLu,Larsson}. However, preliminary studies on massive MIMO systems have for the most part only analyzed its performance under the assumption of perfect hardware \cite{Ngo1,HYang}. The impact of hardware imperfections and nonlinearities on massive MIMO systems has recently been investigated in \cite{Emil1}-\cite{Mollen4}. Although it is well-known that the dynamic power in massive MIMO systems can be scaled down proportional to $\sqrt{M}$, where $M$ denotes the number of BS antennas, the static power consumption at the BS will increase proportionally to ${M}$ \cite{Emil2}. Hence, considering hardware-aware design together with power consumption at the BS seems necessary in realizing practical massive MIMO systems.

Among the various components responsible for power dissipation at the BS, the contribution of analog-to-digital converters (ADCs) is known to be dominant \cite{Bai}. Consequently, the idea of replacing the high-power high-resolution ADCs with power efficient low-resolution ADCs could be a viable approach to address power consumption concerns at the massive MIMO BSs. The impact of utilizing low-resolution ADCs on the spectral efficiency (SE) and energy consumption of massive MIMO systems has been considered in \cite{Li}-\cite{JZhang}. In particular, studies on massive MIMO systems with purely one-bit ADCs show that the high spatial multiplexing gain owing to the use of a large number of antennas is still achievable even with one-bit ADCs \cite{Li,Larsson2}. However, many more antennas with one-bit ADCs (at least 2-2.5 times) are required to attain the same performance as in the high-resolution ADCs case.

One of the main causes of SE degradation in purely one-bit massive MIMO systems is the error due to the coarse quantization that occurs during the channel estimation phase. While at low SNR the loss due to one-bit quantization is only about 2 dB, at higher SNRs performance degrades considerably more and leads to an error floor \cite{Li}. The SE degradation can be reduced by improving the quality of the channel estimation prior to signal detection. One approach for doing so is to exploit so-called mixed-ADC architectures during the channel estimation phase, in which a combination of low- and high-resolution ADCs are used side-by-side. This architecture is depicted in Fig. \ref{system_model_fig}. Mixed-ADC implementations were introduced in \cite{Liang1,Liang2} and their performance was studied from an information theoretic perspective via generalized mutual information.

The basic premise behind the mixed-ADC architecture is to achieve the benefits of conventional massive MIMO systems by just exploiting $N\ll M$ pairs of high-resolution ADCs.
An SE analysis of mixed-ADC massive MIMO systems with maximum ratio combining (MRC) detection for Rayleigh and Rician fading channels was carried out in \cite{Tan} and \cite{JZhang2}, respectively.
The SE and energy efficiency of mixed-ADC systems compared with systems composed of one-bit ADCs was studied in \cite{Hessam2} for MRC detection, and conditions were derived under which each architecture provided the highest SE for a given power consumption. The advantage of using a mixed-ADC architecture in designing Bayes-optimal detectors for MIMO systems with low-resolution ADCs is reported in \cite{TCZhang}. Although the nonlinearity of the quantization process increases the complexity of the optimal detectors, it is shown that adding a small number of high-resolution ADCs to the system allows for less complex detectors with only a slight performance degradation. Moreover, the benefit of using mixed-ADC architectures in massive MIMO relay systems and cloud-RAN deployments is elaborated in \cite{Liu,Park}.

Most existing work in the mixed-ADC massive MIMO literature has assumed either perfect channel state information (CSI) or imperfect CSI with ``round-robin'' training. In the round-robin training approach \cite{Liang1,Liang2,JZhang2}, the training data is repeated several times and the high-resolution ADCs are switched among the RF chains so that every antenna can have a ``clean'' snapshot of the pilots for channel estimation. This obviously requires a larger portion of the coherence interval to be devoted to training rather than data transmission. More precisely, for $M$ antennas and $N$ pairs of high-resolution ADCs, $M/N$ pilot signals are required in the single-user scenario to estimate all $M$ channel coefficients with high-resolution ADCs. This issue is pointed out in \cite{Liang1} for the single user scenario and its impact is taken into account. This training overhead will be exacerbated in the multiuser scenario where orthogonal pilot sequences should be assigned to the users. In this case, the training period becomes $(M/N)\eta$, where $\eta$ represents the length of the pilot sequences (at least as large as the number of user terminals), which could be prohibitively large and may leave little room for data transmission. Hence, it is crucial to account for this fact in any SE analysis of mixed-ADC massive MIMO systems.

In this paper, we examine the channel estimation performance and the resulting uplink SE of mixed-ADC architectures with and without round-robin training, and compare them with implementations that employ uniform ADC quantization across all antennas. The main goals are to determine when, if at all, the benefits of using the round-robin approach with ADC/antenna switching outweigh the cost of increasing the training overhead, and furthermore to examine the question of whether or not one should employ a mixed-ADC architecture in the first place. The contributions of the paper can be summarized as follows.
\begin{itemize}
\item We first present an extension of the round-robin training approach that incorporates both high-resolution and one-bit measurements for the channel estimation. The round-robin training proposed in \cite{Liang1,Liang2,JZhang2} based the channel estimate on only high-resolution observations, assuming that no data was collected from antennas during intervals when they were not connected to the high-resolution ADCs. In contrast, our extension assumes that these antennas collect one-bit observations and combine this data with the high-resolution samples to improve the channel estimation performance.
\item We use the Bussgang decompositon \cite{JJBussgang} to develop a linear minimum mean-squared error (LMMSE) channel estimator based on the combined round-robin measurements and we derive a closed-form expression for the resulting mean-squared error (MSE). We further illustrate the importance of using the Bussgang approach rather than the simpler additive quantization noise model in obtaining the most accurate characterization of the channel estimation performance for round-robin training. The analysis illustrates that the addition of the one-bit observations considerably improves performance at low SNR.
\item We perform a spectral efficiency analysis of the mixed-ADC implementation for the MRC and ZF receivers, and obtain expressions for a lower bound on the SE that takes into account the channel estimation error and the loss of efficiency due to the round-robin training. We compare the resulting SE with that achieved by mixed-ADC implementations that do not switch ADCs among the RF chains, and hence do not use round-robin training. We also compare against the SE for architectures that do not mix the ADC resolution across the array, but instead use uniform resolution with a fixed number of comparators for different array sizes. We show that, depending on the SNR, coherence interval, number of high-resolution ADCs, and the choice of the linear receiver, there are situations where each of the considered approaches shows superior performance. In particular, using uniform low-resolution ADCs is better than a mixed-ADC approach for an interference limited system. On the other hand, a mixed-ADC system, even one with round-robin training, is superior at higher SNRs when zero-forcing is used to reduce the interference.\color{black}
\item We analyze the possible SE improvement that can be achieved by using an antenna selection algorithm that connects the high-resolution ADCs to the subset of antennas with the highest channel gain. We analytically derive the SE performance of the antenna selection algorithm for MRC and numerically study its performance for ZF detection, comparing against the simpler approach of assigning the high-resolution ADCs to an arbitrary fixed subset of the RF chains.
\end{itemize}

In addition to the above contributions, we also discuss some of the issues related to implementing an ADC switch or multiplexer in hardware that allows different ADCs to be assigned to different antennas. We restrict our analysis and numerical examples to a single-carrier flat-fading scenario, although our methodology can be used in a straightforward way to extend the results to frequency-selective fading or multiple-carrier signals ({\em e.g.}, see our prior work in Section~III.B of \cite{Li} for the SE analysis of an all-one-bit ADC system for OFDM and frequency selectivity). The reasons for focusing on the single-carrier flat-fading case are as follows: (1) the mixed-ADC assumption already makes the resulting analytical expressions quite complicated even for the simple flat-fading case, and it would be more difficult to gain insight into the problem if the expressions were further complicated; (2) the original round-robin training idea was proposed in \cite{Liang1} for the single-carrier flat-fading case, and thus we analyze it under the same assumptions; (3) the main conclusions of the paper are based on relative algorithm comparisons for the same set of assumptions, and we expect our general conclusions to remain unchanged if frequency rather than flat fading were considered; and (4) the flat fading case is still of interest in some applications, for example in a micro-cell setting with typical path-length differences of 50-100 m, the coherence bandwidth is between 3-6 MHz, which is not insignificant.

Further assumptions regarding the system model are outlined in the next section. Section \ref{sec:Training_phase} discusses channel estimation using round-robin training, and derives the LMMSE channel estimator that incorporates both the high-resolution and one-bit observations. A discussion of hardware and other practical considerations associated with using a mixed-ADC system with ADC/antenna switching is presented in Section \ref{sec:Practical}. Section \ref{sec:spectral efficiency} then presents the analysis of the spectral efficiency for MRC and ZF receivers based on the imperfect channel state estimates, including an analytical performance characterization of antenna selection and architectures with uniform ADC resolution across the array. A number of numerical studies are then presented in Section \ref{sec:Simulation} to illustrate the relative performance of the algorithms considered.
\color{black}

\emph{Notation}: We use boldface letters to denote vectors, and capitals to denote matrices. The symbols $(.)^*$, $(.)^T$, and $(.)^H$ represent conjugate, transpose, and conjugate transpose, respectively. A circularly-symmetric complex Gaussian (CSCG) random vector with zero mean and covariance matrix ${\mathbf{R}}$ is denoted $\boldsymbol{v}\sim\mathcal{CN}(\mathbf{0},\boldsymbol{\mathbf{R}})$. The symbol $\|.\|$ represents the Euclidean norm. The $K \times K$ identity matrix is denoted by $\mathit{\boldsymbol{I}_K}$ and the expectation operator by $\mathbb{E}\{.\}$. We use $\mathbf{1}_{N}$ to denote the $N\times 1$ vector of all ones, and $\mathrm{diag}\{\mathbf{C}\}$ the diagonal matrix formed from the diagonal elements of the square matrix $\mathbf{C}$. For a complex value, $c=c_R+jc_I$, we define $\mathrm{arcsin}(c)\triangleq\mathrm{arcsin}(c_R)+j\mathrm{arcsin}(c_I)$.

%\emph{Notations}: We use boldface to denote matrices and vectors. The symbols $(.)^*$, $(.)^T$, and $(.)^H$ denotes conjugate, transpose, and conjugate transpose, respectively. The notation $\boldsymbol{v}\sim\mathcal{CN}(\boldsymbol{0},\boldsymbol{\boldsymbol{R}})$ denotes a circularly-symmetric complex Gaussian (CSCG) random vector with zero mean and covariance matrix $\mathit{\boldsymbol{R}}$. We use $\|.\|$ to denote the Euclidean norm, $\mathit{\boldsymbol{I}_K}$ the $K \times K$ identity matrix, $\mathbf{1}_{N}$ the $N\times 1$ vector of all ones, and $\mathrm{diag}\{\mathbf{C}\}$ the diagonal matrix formed from the diagonal elements of the square matrix $\mathbf{C}$. The expectation operator is denoted by $\mathbb{E}\{.\}$. For a complex value, $c=c_R+jc_I$, we define $\mathrm{arcsin}(c)\triangleq\mathrm{arcsin}(c_R)+j\mathrm{arcsin}(c_I)$.

%\vspace{-3mm}

\section{System Model}\label{sec:SYSTEM MODEL}

%\vspace{-1mm}
\begin{figure}
\centering
\includegraphics[width=0.5\textwidth]
{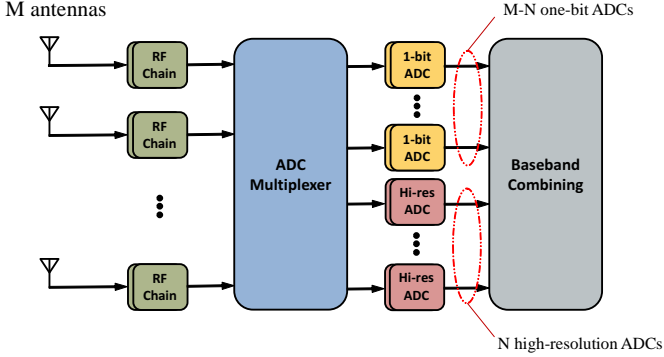}
\caption{Mixed-ADC architecture.}
\label{system_model_fig}
\end{figure}
Consider the uplink of a single-cell multi-user MIMO system consisting of $K$ single-antenna users that send their signals simultaneously to a BS equipped with ${M}$ antennas. Assuming a single-carrier frequency flat channel and symbol-rate sampling
%\footnote{\color{blue}Note that the use of highly nonlinear coarse quantization prevents a wide-band channel from being decomposed to frequency-flat channels via orthogonal-frequency-division-multiplexing (OFDM) in systems that use one-bit ADCs (see \cite{Li} Section III.B and \cite{Liang2}). Hence, in this paper, similar to previous work (\cite{Liang1} and references therein) we focus on flat fading channels and leave analysis of the frequency selective case for future work. 
\color{black}, the $M\times 1$ signal received at the BS from the $K$ users is given by
\begin{equation}\label{channel model}
  \mathit{\boldsymbol{r}}=\sum_{k=1}^{K}{\sqrt{p_k}\boldsymbol{g}_k
  \mathit{{s}_k}} + \mathit{\boldsymbol{n}},
\end{equation}
where $p_k$ represents the average transmission power from the $k$th user, $\boldsymbol{g}_k=\sqrt{\beta_k}\boldsymbol{h}_k$ is the channel vector between the $k$th user and the BS where $\beta_k$ models geometric attenuation and shadow fading, and $\boldsymbol{h}_k\sim\mathcal{CN}\left(\mathbf{0},\boldsymbol{I}_M\right)$ represents the fast fading  and is assumed to be independent of other users' channel vectors. 
The symbol transmitted by the $k$th user is denoted by $s_k$ where $\mathbb{E}\left\{|s_k|^2\right\}=1$ and is drawn from a CSCG codebook independent of the other users.
%$x_k$ is the symbol transmitted from the $k$th user and is drawn from $\mathit{\boldsymbol{x}}\in\mathbb{C}^{K\times1}$ which satisfies $\mathbb{E}\{\mathit{\boldsymbol{x}}\mathit{\boldsymbol{x}}^{H}\}=\boldsymbol{I}_K$, 
Finally, $\mathit{\boldsymbol{n}}\sim\mathcal{CN}\left(\mathbf{0},\sigma_n^2\boldsymbol{I}_M\right)$ denotes additive CSCG receiver noise at the BS.
The assumption of symbol-rate sampling means that the matched filter at the receiver must be implemented in the analog domain. Better performance (e.g., higher rates) could be achieved by oversampling the ADCs, particularly those with one-bit resolution.
\color{black}

We consider a block-fading model with coherence bandwidth $W_c$ and coherence time $T_c$. In this model, each channel remains constant in a coherence interval of length $\mathit{T}=T_cW_c$ symbols and changes independently between different intervals. \color{black} Note that $T$ is a fixed system parameter chosen as the minimum coherence duration of all users.
At the beginning of each coherence interval, the users send their $\eta$-tuple mutually orthogonal pilot sequences ($K\leq\eta\leq T$) to the BS for channel estimation. Denoting the length of the training phase as $\eta_{\mathrm{eff}}$, the remaining $T-\eta_{\mathrm{eff}}$ symbols are dedicated to uplink data transmission. 
%The average transmission powers of the $k$th user during training and data transmission phases are denoted by $p_k^t$ and $p_k^d$, respectively.

%\vspace{-1mm}

\section{Training Phase}\label{sec:Training_phase}
In this section, we investigate the linear minimum mean squared error (LMMSE) channel estimator for different ADC architectures at the BS. In all scenarios, the pilot sequences are drawn from 
 an $\eta\times K$ matrix $\mathbf{\Phi}$, where the $k$th column of $\mathbf{\Phi}$, $\boldsymbol{\phi}_k$, is the $k$th user's pilot sequence and $\mathbf{\Phi}^H\mathbf{\Phi}=\boldsymbol{I}_K$. Therefore, the $M\times\eta$ received signal at the BS before quantization becomes
\begin{equation}\label{training matrix}
  \boldsymbol{X}=\sum_{k=1}^{K}{\sqrt{\eta p_k}\boldsymbol{g}_k\boldsymbol{\phi}_k^T}+\boldsymbol{N},
\end{equation}
where $\boldsymbol{N}$ is an $M\times\eta$ matrix with i.i.d. $\mathcal{CN}(0,\sigma_n^2)$ elements. Since the rows of $\boldsymbol{X}$ are mutually independent due to the assumption of spatially uncorrelated Gaussian channels and noise, we can analyze them separately. \color{black}As a result, we will focus on the $m$th row of $\boldsymbol{X}$ which is
\begin{equation}\label{training matrix row}
	\boldsymbol{x}_m^T=\sum_{k=1}^{K}{\sqrt{\eta p_k}g_{mk}\boldsymbol{\phi}_k^T}+\boldsymbol{n}_m^T, 
\end{equation}
where $g_{mk}$ is the $m$th element of the $k$th user channel vector, $\boldsymbol{g}_k$, and $\boldsymbol{n}_m^T$ is the $m$th row of $\boldsymbol{N}$. Since the analysis is not dependent on $m$, hereafter we drop this subscript and denote the received signal at the $m$th antenna by $\boldsymbol{x}$.
\subsection{ Estimation\ Using\ One-Bit\ Quantized Observations}\
In this subsection, to have a benchmark for comparison purposes, we consider the case in which all antennas at the BS are connected to one-bit ADCs. The received signal $\boldsymbol{x}^T$ after quantization by one-bit ADCs can be written as
\begin{equation}\label{quantized training1}
	\boldsymbol{y}_t^T=\mathcal{Q}\left(\boldsymbol{x}^T\right),
\end{equation}
where the element-wise one-bit quantization operation $\mathcal{Q}(\cdot)$ replaces each input entry with the quantized value $\frac{1}{\sqrt{2}}\left( \pm 1 \pm j\right)$, depending on the sign of the real and imaginary parts. According to the Bussgang decomposition \cite{JJBussgang}, the following linear representation of the quantization can be employed \cite{Li}:  
\begin{equation}\label{quantized training2}
	\mathcal{Q}\left(\boldsymbol{x}^T\right)=\sqrt{\frac{2}{\pi}}\boldsymbol{x}^T\mathbf{D}_{\boldsymbol{x}}^{-\frac{1}{2}}+\boldsymbol{q}_t^T,
\end{equation}
where $\mathbf{D}_{\boldsymbol{x}}=\mathrm{diag}\{\mathbf{C}_{\boldsymbol{x}}\}$ and $\mathbf{C}_{\boldsymbol{x}}$ denotes autocorrelation matrix of $\boldsymbol{x}$, which can be calculated as
\begin{equation}\label{x_m auto}
	\mathbf{C}_{\boldsymbol{x}}=\sum_{k=1}^{K}{\eta p_k\beta_k\boldsymbol{\phi}_k^{*}\boldsymbol{\phi}_k^T}+\sigma_n^2\mathbf{I}_{\eta}.
\end{equation}
In addition, $\boldsymbol{q}_t$ represents quantization noise which is uncorrelated with $\boldsymbol{x}$ and its autocorrelation matrix can be derived based on the \emph{arcsine law} as \cite{Jacovitti}
\begin{equation}\label{quantization noise autocorrelation}
	\mathbf{C}_{\boldsymbol{q}_t}=\frac{2}{\pi}\mathrm{arcsin}\{\mathbf{D}_{\boldsymbol{x}}^{-\frac{1}{2}} \mathbf{C}_{\boldsymbol{x}}  \mathbf{D}_{\boldsymbol{x}}^{-\frac{1}{2}}\}-\frac{2}{\pi}\mathbf{D}_{\boldsymbol{x}}^{-\frac{1}{2}} \mathbf{C}_{\boldsymbol{x}}  \mathbf{D}_{\boldsymbol{x}}^{-\frac{1}{2}}.
\end{equation}

Much of the existing work on massive MIMO systems with low-resolution ADCs employs the simple additive quantization noise model (AQNM) for their analysis \cite{Verenzuela}-\cite{JZhang}, \cite{Tan}-\cite{Park}, \cite{Hessam} which is valid only for low SNRs and does not capture the correlation among the elements of $\boldsymbol{q}_t$,  which turns out to be of crucial importance in our analysis. Hence, we consider the Bussgang decomposition instead and will show its effect on the system performance analysis. Stacking the rows of (\ref{quantized training2}) into a matrix, the one-bit quantized observation at the BS becomes
\begin{equation}\label{matrix quantized training}
	\mathbf{Y}=\sqrt{\frac{2}{\pi}}\mathbf{X}\mathbf{D}_{\boldsymbol{x}}^{-\frac{1}{2}}+\mathbf{Q},
\end{equation}
where $\mathbf{Q}$ is an $M\times\eta$ matrix whose $m$th row is $\boldsymbol{q}_t^T$.
The LMMSE estimate of the channel $\boldsymbol{G}=[\boldsymbol{g}_1,...,\boldsymbol{g}_K]$ based on just one-bit quantized observations (\ref{matrix quantized training}) is given in the following theorem.
\begin{theo}\label{theo1}
The LMMSE estimate of the $k$-th user channel, $\boldsymbol{g}_k$, given the one-bit quantized observations $\mathbf{Y}$ is \cite{Li}
\begin{equation}\label{channel estimate}
  \hat{\boldsymbol{g}}_k=\frac{\beta_k}{\beta_k+\sigma^2_{w_k}}\sqrt{\frac{1}{\eta p_k}}\mathbf{Y}\bar{\boldsymbol{\phi}}_k^{*}
  ,
\end{equation}
where 
\begin{equation}\label{effective pilot}
	\bar{\boldsymbol{\phi}}_k\triangleq\sqrt{\frac{\pi}{2}}\mathbf{D}_{\boldsymbol{x}}^{\frac{1}{2}}\boldsymbol{\phi}_k
\end{equation}
\begin{equation}\label{error variance}
	\sigma^2_{w_k}=
	\frac{1}{\eta p_k}\left(\sigma_n^2+\bar{\boldsymbol{\phi}}_k^T\mathbf{C}_{\boldsymbol{q}_t}\bar{\boldsymbol{\phi}}_k^{*}\right).
	\end{equation}
	Define the channel estimation error $\boldsymbol{\varepsilon}\triangleq \hat{\boldsymbol{g}}_k-\boldsymbol{g}_k$. Then we have
\begin{equation}\label{ghat variance}
  \sigma_{\hat{g}_k}^{2}=\frac{\beta_k^2}{\beta_k+\sigma^2_{w_k}}~~~\text{and}~~~\sigma_{{\varepsilon}_k}^{2}=\frac{\sigma^2_{w_k}\beta_k}{\beta_k+\sigma^2_{w_k}},
\end{equation}
where $\sigma_{\hat{g}_k}^{2}$ and $\sigma_{{\varepsilon}_k}^{2}$ are the variances of the independent zero-mean elements of $\hat{\boldsymbol{g}}_k$ and $\boldsymbol{\varepsilon}$, respectively.
\end{theo}
From Theorem \ref{theo1}, it is apparent that in the channel estimation analysis of massive MIMO systems with one-bit ADCs, the estimation error is directly affected not only by the inner product of the pilot sequences, but also by their outer product as well \cite{Li}. 
To get insight into the impact of the one-bit quantization on the channel estimation, in the next corollary we adopt the statistics-aware power control policy proposed in \cite{Emil Bj}. Apart from its practical advantages, this policy is especially suitable specially for one-bit ADCs since it avoids near-far blockage and hence strong interference. Moreover, this power control approach also leads to simple expressions and provides analytical convenience for our derivation in Section VI.  Although not the focus of this paper, we note that in general a massive MIMO system employing a mixed-ADC architecture will be more resilient than an all one-bit implementation to the near-far effect and jamming. This is an interesting topic for further study.
\color{black}

\begin{cor}\label{corol1}
For the case in which power control is performed, i.e., $p_k=\frac{p}{\beta_k}$ for some fixed value $p$ and for $k\in\mathcal{K}=\{1,\cdots,K\}$, the number of users is equal to the length of pilot sequences, i.e., $\eta=K$, and the pilot matrix satisfies $\mathbf{\Phi}\mathbf{\Phi}^H=\mathbf{I}_K$, we have
\begin{equation}\label{cor1_1}
	\mathbf{C}_{\boldsymbol{x}}=\left(Kp+\sigma_n^2\right)\mathbf{I}_K=\mathbf{D}_{\boldsymbol{x}}
\end{equation}	 
\begin{equation}\label{cor1_2}
	\mathbf{C}_{\boldsymbol{q}_t}=\left(1-\frac{2}{\pi}\right)\mathbf{I}_K,
\end{equation}
which yields
\begin{equation}\label{cor1_3_1}
	\sigma_{\hat{g}_k}^{2}=\frac{2}{\pi}\frac{\beta_k}{1+\frac{\sigma_n^2}{Kp}}
\end{equation}
\begin{equation}\label{cor1_3_2}
\sigma_{{\varepsilon}_k}^{2}=\frac{\left(\frac{Kp}{\sigma^2_{n}}\left(1-\frac{2}{\pi}\right)+1\right)\beta_k}{1+\frac{Kp}{\sigma^2_{n}}}.	
\end{equation}

\end{cor}
Corrollary \ref{corol1} states conditions under which $\mathbf{C}_{\boldsymbol{q}_t}$ is diagonal. In addition, it is evident that the channel estimation suffers from an error floor at high SNRs.
%\vspace{-1mm}

\subsection{ Channel\ Estimation\ with\ Few\ Full\ Resolution\ ADCs}\
 Channel estimation with coarse observations suffers from large errors especially in the high SNR regime. On the other hand, while estimating all channels using high-resolution ADCs is desirable, the resulting power consumption burden makes this approach practically infeasible. This motivates the use of a mixed-ADC architecture for channel estimation to eliminate the large estimation error caused by one-bit quantization while keeping the power consumption penalty at an acceptable level. In the approach described in \cite{Liang1,Liang2,JZhang2} , $N\ll M$ pairs of high-resolution ADCs are deployed and switched between different antennas during different transmission intervals in an approach referred to as ``round-robin'' training. In this approach, 
\begin{figure}
\centering
\includegraphics[width=0.48\textwidth]
{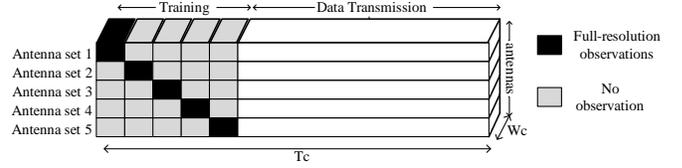}%{fig1_eps}
\caption{Transmission protocol for estimation using full-resolution observations.}
\label{fig11_model}
\end{figure}
the $M$ BS antennas are grouped into $M/N$ sets\footnote{We assume $M/N$ is an integer throughout the paper.}. In the first training sub-interval, users send their mutually orthogonal pilots to the BS while the $N$ high-resolution ADC pairs are connected to the first set of $N$ antennas. After receiving the pilot symbols from all users in the $\eta$-symbol-length training sub-interval, the high-resolution ADCs are switched to the next set of antennas and so on. In this manner, after $(M/N)\eta$ pilot transmissions ($M/N$ sub-intervals), we can estimate each channel based on observations with only high-resolution ADCs.
%\footnote{\color{blue}Currently available off-the-shelf RF switches have switching times on the order of a few nanoseconds. For example, an Analog Devices AD8109 8x8 analog crosspoint switch has a switching time of 25 ns \cite{ADC}. Note that in our application, the switching only needs to take place $M/N$ times at intervals of $\eta\ge K$ symbols during each channel coherence interval. If guard intervals are inserted to account for the switching transients, then the impact on the overall spectral efficiency should be inconsequential. For example, if $M/N=10$, the required extra guard time would be 250 ns, which is insignificant compared to channel coherence times even at millimeter wave frequencies. Using a custom designed hardware, one could expect to achieve even faster performance than that for off-the-shelf components, so the switching transients should not impact the analysis.\color{black}}.
This round-robin channel estimation protocol is illustrated in Fig. \ref{fig11_model} for a mixed-ADC system with $M/N=5$. 

Stacking all $N\times\eta$ full-resolution observations into an $M\times\eta$ matrix, $\mathbf{X}$, the LMMSE estimate of the $k$-th user channel, $\boldsymbol{g}_k$, is \cite{Ngo1}
\begin{equation}\label{channel estimate2}
  \hat{\boldsymbol{g}}_k=\frac{1}{1+\frac{\sigma_n^2}{\eta p_k\beta_k}}\frac{1}{\sqrt{\eta p_k}}\mathbf{X}\boldsymbol{\phi}_k^{*},
\end{equation}
and the resulting variances of the channel estimate and the error are given respectively by
\begin{equation}\label{ghat variance2}
  \sigma_{\hat{g}_k}^{2}=\frac{\beta_k}{1+\frac{\sigma_n^2}{\eta p_k\beta_k}}~~~\text{and}~~~\sigma_{{\varepsilon}_k}^{2}=\frac{\beta_k}{1+\frac{\eta p_k\beta_k}{\sigma^2_{n}}}.
\end{equation}
%\end{theo}

Eq. (\ref{ghat variance2}) states that by employing only $N$ pairs of high-resolution ADCs and by expending a larger portion of the coherence interval for channel estimation, the channel can be estimated with the same precision as that achieved by conventional high-resolution ADC massive MIMO systems. However, this comes at the high cost of repeating the training data $M/N$ times, which can significantly reduce the time available for data transmission. Indeed, we will see later that in some cases, a mixed-ADC implementation with round-robin training achieves a lower SE than a system with all one-bit ADCs because of the long training interval (even with the improvements we propose below for the round-robin method). However, we will also see that there are other situations for which the mixed-ADC round-robin method provides a large gain in SE. The primary goal of this paper is to elucidate under what conditions these and other competing approaches provide the best performance. \color{black}

Before analyzing the tradeoff between the gain (lower channel estimation error) and cost (longer training period) of the round-robin approach, in the next subsection we propose channel estimation based on the use of both full-resolution and one-bit data received by the BS in order to further improve the performance of the mixed-ADC architecture with round-robin channel estimation. To our knowledge, this approach has not been considered in prior work on mixed-ADC massive MIMO.

%\vspace{-1mm}

\subsection{Estimation\ Using\ Joint\ Full-Resolution/One-Bit Observations}\
 While channel estimation performance based on coarsely quantized observations suffers from large errors in the high SNR regime, it provides reasonable performance for low SNRs. Hence, in this subsection we consider joint channel estimation based on observations from both high-resolution and one-bit ADCs to further improve the channel estimation accuracy. Unlike the previous subsection in which the one-bit ADCs were not employed, here we incorporate their coarse observations into the channel estimation procedure. The protocol for this method is illustrated in Fig. \ref{fig22} for a mixed-ADC system with $M/N=5$. 
\begin{figure}
\centering
\includegraphics[width=0.48\textwidth]
{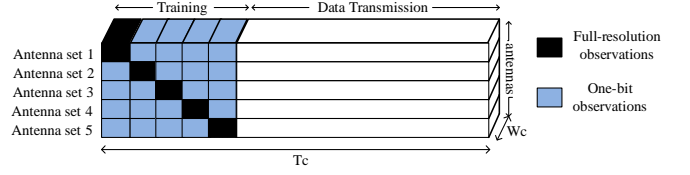}%{fig2_eps}
\caption{Transmission protocol for estimation using full-resolution/one-bit observations.}
\label{fig22}
\end{figure}
It can be seen that, in addition to one set of full-resolution observations for each antenna, there are $(M/N)-1$ sets of one-bit observations which are also taken into account for channel estimation. The next theorem characterizes the performance of this approach.
\begin{theo}\label{theo3}
Stacking all $N\times\eta$ full-resolution observations into an $M\times\eta$ matrix, $\mathbf{X}$, and all $(M/N)-1$ $N\times\eta$ one-bit quantized observations into $M\times\eta$ matrices, $\mathbf{Y}_t$, $t\in\mathcal{T}=\{1,...,M/N-1\}$, the LMMSE estimate of the $k$-th user channel, $\boldsymbol{g}_k$, is
\begin{equation}\label{channel estimate3}
  \hat{\boldsymbol{g}}_k=\sqrt{\frac{1}{\eta p_k}}\left(w_{\infty_k}\mathbf{X}\boldsymbol{\phi}_k^{*} + w_{1_k}\sum_{t=1}^{\frac{M}{N}-1}{\mathbf{Y}_t\bar{\boldsymbol{\phi}}_k^{*}}\right),
\end{equation}
where 
\begin{equation}\label{weight1}
	w_{\infty_k}=\frac{\frac{\eta p_k}{\sigma_n^2}}{\frac{1}{\beta_k}+\frac{\eta p_k}{\sigma_n^2}+\varsigma_k(p_k)}
\end{equation}
\begin{equation}\label{weight2}
	w_{1_k}=\frac{\left(\frac{M}{N}-1\right)^{-1}\varsigma_k(p_k)}{\frac{1}{\beta_k}+\frac{\eta p_k}{\sigma_n^2}+\varsigma_k(p_k)}
\end{equation}
\begin{equation}\label{varsigma_def}
	\varsigma_k(p_k)=\frac{\left(\frac{M}{N}-1\right)}{\sigma^2_{w_k}+\left(\frac{M}{N}-2\right)\varrho_k}
\end{equation}
\begin{equation}\label{sigma_w}
	\sigma^2_{w_k}=
	\frac{1}{\eta p_k}\left(\sigma_n^2+\bar{\boldsymbol{\phi}}_k^T\mathbf{C}_{\boldsymbol{q}_t}\bar{\boldsymbol{\phi}}_k^{*}\right)
\end{equation}
\begin{equation}\label{covariance}
	\varrho_k=\frac{1}{\eta p_k}\bar{\boldsymbol{\phi}}_k^T\bar{\mathbf{C}}_{\boldsymbol{q}_t}\bar{\boldsymbol{\phi}}_k^{*}
\end{equation}
\begin{equation}\label{covar_new}
	\bar{\mathbf{C}}_{\boldsymbol{q}_t}=\frac{2}{\pi}\mathrm{arcsin}\{\bar{\mathbf{D}}_{\boldsymbol{x}}^{-\frac{1}{2}} \bar{\mathbf{C}}_{\boldsymbol{x}}  \bar{\mathbf{D}}_{\boldsymbol{x}}^{-\frac{1}{2}}\}-\frac{2}{\pi}\bar{\mathbf{D}}_{\boldsymbol{x}}^{-\frac{1}{2}} \bar{\mathbf{C}}_{\boldsymbol{x}}  \bar{\mathbf{D}}_{\boldsymbol{x}}^{-\frac{1}{2}}
\end{equation}
\begin{equation}\label{covar_new2}
	\bar{\mathbf{C}}_{\boldsymbol{x}}=\sum_{k=1}^{K}{\eta p_k\beta_k\boldsymbol{\phi}_k^{*}\boldsymbol{\phi}_k^T}
\end{equation}
\begin{equation}
\bar{\mathbf{D}}_{\boldsymbol{x}}=\mathrm{diag}\{\bar{\mathbf{C}}_{\boldsymbol{x}}\}.
\end{equation}
This approach yields the following variances for the channel estimate and the estimation error, respectively:
\begin{equation}\label{ghat variance3}
  \sigma_{\hat{g}_k}^{2}=\frac{\frac{\eta p_k}{\sigma_n^2}+\varsigma_k(p_k)}{\frac{1}{\beta_k}+\frac{\eta p_k}{\sigma_n^2}+\varsigma_k(p_k)}\beta_k
\end{equation}
\begin{equation}\label{error variance3}
  \sigma_{\varepsilon_k}^{2}=\frac{1}{\frac{1}{\beta_k}+\frac{\eta p_k}{\sigma_n^2}+\varsigma_k(p_k)}.
\end{equation}
\end{theo}
\color{black}
\begin{proof}
See Appendix A.	
\end{proof}

%%\vspace{-3mm}
%\begin{rem}\label{rem1}
	Theorem \ref{theo3} demonstrates the optimal approach for combining the observations from high-resolution and one-bit ADCs. In addition, this highlights the importance of considering the correlation among the one-bit observations in the analysis of mixed-ADC channel estimation, something that could not be addressed by the widely-used AQNM approach. More precisely, it can be seen that the impact of joint high-resolution/one-bit channel estimation is manifested in the variance of the channel estimation error by the term $\varsigma_k(p_k)$. To see this, assume that the correlation among one-bit observations in different training sub-intervals is ignored (as would be the case with the AQNM approach). As shown in the appendix, this is equivalent to setting $\varrho_k = 0$ in (\ref{covariance}). Under this assumption, $\varsigma_k(p_k)$ becomes
\begin{equation}\label{varsigma_no_cor}
	\varsigma_{k_0}(p_k)=\frac{\left(\frac{M}{N}-1\right)}{\sigma^2_{w_k}}>\varsigma_{k}(p_k),
\end{equation}
and thus, $\sigma_{\varepsilon_{k}}^{2}>\sigma_{\varepsilon_{k_0}}^{2}$ where $\sigma_{\varepsilon_{k_0}}^{2}$ denotes the estimation error for $\varrho_k=0$. Consequently, the AQNM model yields an overly optimistic assessment of the channel estimation error compared with the more accurate Bussgang analysis. We will see below that the impact of the AQNM approximation is significant for mixed-ADC channel estimation.
%\end{rem}

The next corollary provides insight into the impact of the system parameters on the joint high-resolution/one-bit LMMSE estimation.\color{black}
\begin{cor}\label{corol2}
For the case in which power control is performed, i.e., $p_k=\frac{p}{\beta_k}$ for $k\in\mathcal{K}$, the number of users is equal to the length of pilot sequences, i.e., $\eta=K$, and the pilot matrix satisfies $\mathbf{\Phi}\mathbf{\Phi}^H=\mathbf{I}_K$, we have
\begin{equation}\label{cor1_1}
	\bar{\mathbf{C}}_{\boldsymbol{x}}=Kp\mathbf{I}_K=\bar{\mathbf{D}}_{\boldsymbol{x}},
\end{equation}	
and 
\begin{equation}\label{cor1_2}
	\bar{\mathbf{C}}_{\boldsymbol{q}_t}=\left(1-\frac{2}{\pi}\right)\mathbf{I}_K,
\end{equation}
which yields
\begin{equation}\label{cor1_3_1}
	\sigma_{\hat{g}_k}^{2}=\frac{\frac{Kp}{\sigma_n^2}+\varsigma(p)}{1+\frac{Kp}{\sigma_n^2}+\varsigma(p)}\beta_k~~and~~\sigma_{{\varepsilon}_k}^{2}=\frac{1}{1+\frac{Kp}{\sigma_n^2}+\varsigma(p)}\beta_k,
\end{equation}
where 
\begin{equation}\label{cor1_3_2}
	\varsigma(p)= \frac{\left(\frac{M}{N}-1\right)}{\frac{\pi}{2}\frac{\sigma_n^2}{Kp}+\left(\frac{\pi}{2}-1\right)\left(\frac{M}{N}-1\right)}.
\end{equation}
In addition,
\begin{equation}\label{cor1_3_3}
	w_{\infty}=\frac{\frac{Kp}{\sigma_n^2}}{1+\frac{Kp}{\sigma_n^2}+\varsigma(p)}~~~and~~~w_{1}=\frac{\left(\frac{M}{N}-1\right)^{-1}\varsigma(p)}{1+\frac{Kp}{\sigma_n^2}+\varsigma(p)},
\end{equation}
where $w_{\infty}$ and $w_{1}$ denote the weights of the high-resolution and one-bit observations in the LMMSE estimation, respectively.
\end{cor}
Corallary \ref{corol2} states that in contrast to Theorem 1 where the correlation among one-bit observations within each training sub-interval can be eliminated by carefully selecting the system parameters as in Corollary \ref{corol1}, we cannot overcome the correlation among one-bit observations from different training sub-intervals. This phenomenon makes the addition of the one-bit observations less useful especially in the high SNR regime. For instance, in the asymptotic case, as the $\text{SNR}=\frac{p}{\sigma_n^2}$ goes to infinity, we have
\begin{equation}\label{after_cor_2_1} 
\varsigma\longrightarrow\frac{1}{\frac{\pi}{2}-1}	,
\end{equation}
\begin{equation}\label{after_cor_2_2} 
w_{\infty}\longrightarrow 1,~w_1\longrightarrow 0.
\end{equation}
It is apparent from (\ref{after_cor_2_1}) that in the asymptotic regime $\varsigma$ tends to a finite value and also is independent of $M/N$. Moreover, (\ref{after_cor_2_2}) implies that the optimal approach for high SNRs is to estimate the channel based solely on the high-resolution observations.  
%From (\ref{cor1_3_1}), it can be seen that the impact of joint high-resolution/one-bit channel estimation is manifested in the variance of the channel estimation error by the term $\varsigma(p)$. It is worthwhile to note that if one can overcome the correlation among one-bit observations in different training sub-intervals, i.e., $\varrho_k=0$, $\varsigma(p)$ becomes
%\begin{equation}\label{no_cor_alpha}
%	\varsigma_0(p)=\frac{\left(\frac{M}{N}-1\right)}{\frac{\pi}{2}\frac{\sigma_n^2}{Kp}+\left(\frac{\pi}{2}-1\right)},
%\end{equation}
\begin{figure}
\centering
\includegraphics[width=0.5\textwidth]
{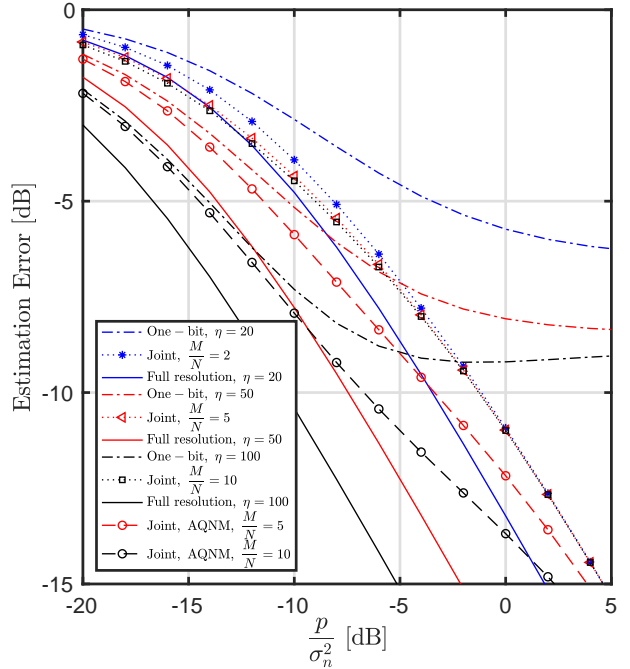}
%{estimation_error_LS_dB}
\caption{Channel estimation error $\sigma_{{\varepsilon}_k}^{2}/\beta_k$ versus ${p}/{\sigma_n^2}$.}
\label{fig1}
\end{figure}

The error for the three channel estimation approaches in Eqs. (\ref{ghat variance}), (\ref{ghat variance2}), and (\ref{error variance3})  is depicted in Fig. \ref{fig1} for a case with $M=100$ antennas, $K=10$ users, and various numbers of high-resolution ADCs, $N$ and training lengths $\eta$. The label ``Joint'' refers to round-robin channel estimation that includes the one-bit observations as described in the previous section, "Full resolution" indicates the performance achieved using a full array of high-resolution ADCs, and ``One-bit'' refers to the performance of an all-one-bit architecture. We also plot the performance predicted for the Joint approach based on the AQNM analysis, which ignores the correlation among the one-bit observations. We see that the AQNM-based analysis yields an overly optimistic prediction for the channel estimation error. In particular, unlike AQNM, the more accurate Bussgang analysis shows that channel estimation with all an one-bit BS actually outperforms the Joint method for low SNRs, a critical observation in analyzing whether or not a mixed-ADC implementation makes sense.  However, we see that the mixed-ADC architecture eventually overcomes the error floor of the all one-bit system for high SNRs and in such cases can reduce the estimation error dramatically. Fig. \ref{fig1} focuses on channel estimation performance, but does not reflect the full impact of the round-robin training on the overall system spectral efficiency, since reducing $N$ increases the amount of training required by the round-robin method. This will be taken into account when we analyze the SE in Section \ref{sec:spectral efficiency}.

\section{Practical Considerations}\label{sec:Practical}
The improvement in channel estimation performance provided by the round-robin training clearly comes at the expense of a significantly increased training overhead. For example, consider a simple worst-case example with a 400 Hz Doppler spread in a narrowband channel of 400 kHz bandwidth; in this case, the coherence time is roughly 1000 symbols. For higher bandwidths or smaller cells with lower mobility, the coherence time can easily approach 10,000 symbols or more. A mixed-ADC array of 128 antennas with 16 high-resolution ADCs would require repeating the pilots 8 times, which for 20 users would amount to 160 symbols, or 16\% of the coherence time when $T=1000$ symbols. This is a relatively high price to pay, and as we will see later, in many instances the resulting loss in SE cannot be offset by the improved channel estimate. However, we will also see that on the other hand, there are other situations where the opposite is true, where the round-robin method leads to significant gains in SE even taking the training overhead into account.

Besides the extra training overhead, the round-robin method has the disadvantage of requiring extra RF switching or multiplexing hardware prior to the ADCs, as shown in Fig.~ \ref{system_model_fig}. It is unlikely that a single large $M\times M$ multiplexer would be used for this purpose, since complete flexibility in assigning a given high-resolution ADC to any possible antenna is not needed. A more likely architecture would employ a bank of smaller multiplexers that allows one high-resolution ADC to be switched among a smaller subarray of antennas, ensuring that each RF chain has access to high-resolution training data during one of the round-robin intervals. Such an approach is similar to the simplified ``subarray switching'' schemes proposed for antenna selection in massive MIMO \cite{Garcia}-\cite{Gao2}. In an interesting earlier example, a large $108\times 108$ multiplexer chipset for a local area network application was developed in \cite{Lefevre}, composed of several $36\times 36$ differential crosspoint ASIC switches that consume less than 100 mW each, with a bandwidth of 140 MHz and a 0 dB insertion loss.

In the 20 years since \cite{Lefevre}, RF switch technology has advanced considerably. For the example discussed above involving a 128-element array with 16 high-resolution ADCs and 112 one-bit ADCs, the multiplexing could be achieved using 16 $8\times 8$ analog switches arranged in parallel. Consider the Analog Devices ADV3228 $8\times 8$ crosspoint switch as an example of an off-the-shelf component for such an architecture\footnote{See http://www.analog.com/en/products/switches-multiplexers/buffered-analog-crosspoint-switches/adv3228.html$\#$product-overview for product details.\color{black}}. The ADV3228 has a 750 MHz bandwidth, a switching time of 15 ns, and a power consumption of 500 mW, which is similar to that of an 8-bit ADC (for example, see Texas Instruments' ADC08B200 8-bit 200 MS/s ADC\footnote{http://www.ti.com/product/ADC08B200/technicaldocuments.\color{black}}). Since the switches can be implemented at a lower intermediate frequency prior to the I-Q demodulation, only one per subarray is required, and thus the total power consumption of the switches would be less than half that of the ADCs.

Note that for the vast majority of the coherence time, the switch is idle. To accommodate the round-robin training, the switches only need to be operated $\frac{M}{N}-1$ times, once for every repetition of the training data. This reduces the actual power consumption to below the specification, and further reduces the impact of the additional training. Short guard intervals would need to be inserted between the training intervals to account for
the switching transients, but these will typically not impact the SE. For the example discussed above with 128 antennas and 8 switches, 7 switching events are required for a total switching time of 105 ns, which is insignificant compared to the coherence time of 2.5 ms at a 400 Hz Doppler.

The insertion loss of the analog switches would also have to be taken into account in an actual implementation, since this will directly reduce the overall SNR of the received signals. Harmonic interference due to nonlinearities in the switch are likely not an issue; for example, the specifications for a Texas Instruments switch (LMH6583) similar to the ADV3228 indicate that the power of the second and third harmonic distortions were -76 dBc. Furthermore, it has been shown that the use of signal combining with a massive antenna array provides significant robustness to such nonlinearities and other hardware imperfections \cite{Emil1}-\cite{Mollen4}.

\color{black} 
\section{Spectral Efficiency}\label{sec:spectral efficiency}
Although channel estimation with a mixed-ADC architecture using round-robin training can substantially improve the channel estimation accuracy, it requires a longer training interval and, therefore, leaves less room for data transmission in each coherence interval. More precisely, $(M/N)\eta$ symbol transmissions are required for round-robin channel estimation which could be large when the number of high-resolution ADCs, $N$, is small\footnote{Note that in designing a mixed-ADC system with round-robin channel training, one should
consider the ratio $M/N$ in scaling the system instead of just increasing the number of antennas $M$. In particular, increasing the number of BS antennas requires increasing of the high-resolution
ADCs,
$N$, as well.\color{black}}. Despite losing a portion of the coherence interval for channel estimation due to the mixed-ADC architecture, the improvement in the signal-to-quantization-interference-and-noise ratio (SQINR) can be significant owing to more accurate channel estimation, and thus a higher rate would be expected during this shorter data transmission period. In this section, we study this system performance trade-off in terms of spectral efficiency for the three mentioned channel estimation approaches.

In the data transmission phase, all users simultaneously send their data symbols to the BS. To begin, assume the antennas are ordered so that the last $N$ antennas are connected to high-resolution ADCs in this phase. A more thoughtful assignment of the high-resolution ADCs will be considered below. From equation (\ref{channel model}), and based on the Bussgang decomposition, the received signal at the BS after one-bit quantization is
\begin{equation}\label{received data}
	\boldsymbol{y}_d =
	{\begin{bmatrix}
        \sqrt{\frac{2}{\pi}}\bar{\mathbf{D}}^{-\frac{1}{2}} & \mathbf{0} \\
        \mathbf{0} & \mathbf{I}_{N}
\end{bmatrix}}
\boldsymbol{r}+\underbrace{\begin{bmatrix}
        \bar{\boldsymbol{q}}_{d} \\
        \mathbf{0}
\end{bmatrix}}_{{\boldsymbol{q}}_{d}}
\end{equation}
\begin{equation}\label{diag_correlaion_data}
	\bar{\mathbf{D}}=\mathrm{diag}\{\mathbf{C}_{\boldsymbol{r}}\}
\end{equation}
\begin{equation}\label{correlaion_data}
	\mathbf{C}_{\boldsymbol{r}}=\sum_{k=1}^{K}{p_k\bar{\boldsymbol{g}}_k\bar{\boldsymbol{g}}_k^H}+\sigma_n^2\mathbf{I}_{M-N},
\end{equation}
where $\bar{\boldsymbol{g}}_k$ denotes the $M-N$ elements of $\boldsymbol{g}_k$ corresponding to the $M-N$ one-bit ADCs and ${\bar{\boldsymbol{q}}_{d}}$ is the $(M-N)\times 1$ quantization noise in the data transmission phase.  It is apparent that the covariance matrix in (\ref{correlaion_data}) is not diagonal which makes analytical tractability difficult. However, by adopting statistics-aware power control \cite{Emil Bj}\color{black}, i.e., $p_k=\frac{p}{\beta_k}$, and assuming that the number of users is relatively large (typical for massive MIMO systems), channel hardening occurs \cite{Li}, and (\ref{correlaion_data}) can be approximated as 
\begin{equation}\label{correlaion_data_approx}
	\mathbf{C}_{\boldsymbol{r}}\cong \left(Kp+\sigma_n^2\right)\mathbf{I}_{M-N}=\bar{\mathbf{D}}.
\end{equation}
As a result, according to the arcsine law (see (\ref{quantization noise autocorrelation})), the covariance matrix of the quantization noise in the data transmission phase becomes $\mathbf{C}_{\bar{\boldsymbol{q}}_d}\cong (1-2/\pi)\mathbf{I}_{M-N}$ and (\ref{received data}) simplifies to
\begin{equation}\label{received data_simplified}
	\boldsymbol{y}_d\cong \mathbf{A}\left(\sum_{k=1}^{K}{\sqrt{p}\boldsymbol{h}_ks_k}+\boldsymbol{n}\right)+\boldsymbol{q}_{d}
\end{equation}
\[
\mathbf{A}={\begin{bmatrix}
        \alpha\mathbf{I}_{M-N} & \mathbf{0} \\
        \mathbf{0} & \mathbf{I}_{N}
\end{bmatrix}},
\]
where $\alpha\triangleq\sqrt{\frac{2}{\pi}\frac{1}{\left(Kp+\sigma_n^2\right)}}$.

For data detection, the BS selects a linear receiver ${\boldsymbol{W}}\in\mathbb{C}^{M\times K}$ as a function of the channel estimate. Note that the quantization model considered in (\ref{quantized training1}) and (\ref{quantized training2}) does not preserve the power of the input of the quantizer since the power of the output is forced to be ${1}$. Thus we premultiply the received signal as follows to offset this effect:
%\begin{equation}\label{whitening}
%\bar{\mathbf{W}}=\mathbf{A}^{-1}\mathbf{W}.	
%\end{equation}
\begin{equation}\label{received signal}
  {\hat{\boldsymbol{y}}_d}=\mathbf{A}^{-1}
  \boldsymbol{y}_{d}.
\end{equation}
By employing the linear detector $\mathbf{W}$, the resulting signal at the BS is
\begin{equation}\label{received signal}
  \mathit{\hat{\boldsymbol{s}}}={{\boldsymbol{W}}}^H
  \hat{\boldsymbol{y}}_{d}.
\end{equation}
Thus, the $k$th element of $\hat{\mathit{\boldsymbol{s}}}$ is
\begin{multline}\label{kth user received signal}
  \hat{s}_k=\sqrt{p}{{\mathit{\boldsymbol{w}}}_{k}^{H}{\mathit{\boldsymbol{h}}}_{k}}
  s_k
  +\sqrt{p}\sum_{i=1,i\ne k}^{K}{{{\mathit{\boldsymbol{w}}}_{k}^{H}{\mathit{\boldsymbol{h}}}_{i}}s_i}\\
  +{{\mathit{\boldsymbol{w}}}_k^{H}\mathit{\boldsymbol{n}}}
  +{\mathit{\boldsymbol{w}}}_k^{H}\mathbf{A}^{-1}{\mathit{\boldsymbol{q}_d}},
\end{multline}
where ${\mathit{\boldsymbol{w}}}_k$ is the $k$th column of ${\boldsymbol{W}}$. We assume the BS treats ${\mathit{\boldsymbol{w}}}_{k}^{H}{\mathit{\boldsymbol{h}}}_{k}$ as the gain of the desired signal and the other terms of (\ref{kth user received signal}) as Gaussian noise when decoding the signal\footnote{Note that in general, the quantization noise is not Gaussian. However, to derive a lower bound for the SE, we assume it is Gaussian with covariance $\mathbf{C}_{\boldsymbol{q}_d}$.}.
Consequently, we can use the classical bounding technique of \cite{Emil Bj} to derive an approximation for the ergodic achievable SE at the $k$th user as\color{black}
\begin{equation}\label{SE definition}
  \mathcal{S}_k=\mathcal{R}\left(\mathrm{SQINR}_k\right),%\mathcal{C}\left(\mathit{SINR}_k\right),
\end{equation}
where the effective $\text{SQINR}_k$ is defined by (\ref{uplink achievable rate}) at the top of the next page, and $\mathcal{R}\left(\theta\right)\triangleq\left(1-\eta_{\mathrm{eff}}/T\right)\mathrm{log}_2\left(1+\theta\right)$ where $\eta_{\mathrm{eff}}$ represents the training duration which is $\eta$ and $\left(M/N\right)\eta$ for the pure one-bit and mixed-ADC architectures, respectively. 
\begin{figure*}[!t]
\begin{equation}\label{uplink achievable rate}
%\resizebox{.9\hsize}{!}
\mathrm{SQINR}_k=
\frac{p\left|{\mathbb{E}\left\{{{\mathit{\boldsymbol{w}}}_{k}^{H}{\mathit{\boldsymbol{h}}}_{k}}\right\}}\right|^{2}}{p\sum_{i=1}^{K}{\mathbb{E}\left\{\left|
  {\mathit{\boldsymbol{w}}}_{k}^{H}{\mathit{\boldsymbol{h}}}_{i}\right|^{2}\right\}}
-p\left|\mathbb{E}\left\{{{\mathit{\boldsymbol{w}}}_{k}^{H}{\mathit{\boldsymbol{h}}}_{k}}\right\}\right|^{2}
+\sigma_n^2\mathbb{E}\left\{{\|{\boldsymbol{w}}_k\|^2}\right\}+\alpha^{-2}\mathbb{E}\left\{{{\boldsymbol{w}}}_{k}^{H}\mathbf{C}_{\boldsymbol{q}_d}{\boldsymbol{w}}_{k}\right\}
}
\end{equation}
\hrulefill
    \vspace*{4pt}
    \end{figure*}
\subsection{ MRC\ Detection}\label{sec:MRC_det}

\subsubsection{ Random\ Mixed-ADC\ Detection}\label{sec:Random_Mix}

In this subsection, we consider the case in which the high-resolution ADCs are connected to an arbitrary set of $N$ antennas. 
Denoting the estimate of the channel by $\hat{\mathbf{H}}=[\hat{\boldsymbol{h}}_1,...,\hat{\boldsymbol{h}}_K]$, setting $\mathbf{W}=\hat{\mathbf{H}}$, and following the same reasoning as in \cite{Li}, the SE of the mixed-ADC architecture with MRC detection can be derived as
\begin{equation}\label{SE formulaMRC_mix}
	\mathcal{S}_k^{\mathrm{MRC}}=\mathcal{R}\left(\frac{pM\sigma_{\hat{h}}^{2}}{pK+\sigma_n^2+\frac{\left(1-\frac{2}{\pi}\right)}{\alpha^2}\left(1-\frac{N}{M}\right)}\right),
\end{equation}
where the channel estimate variance $\sigma_{\hat{h}}^{2}=\sigma_{\hat{g}_k}^{2}/\beta_k$ depends on the estimation approach as denoted in (\ref{ghat variance}), (\ref{ghat variance2}), and (\ref{ghat variance3}). 
%ZFZFZFZFFZFZFZFZFZFZFZFZFZFZZF
%Similarly, by setting $\mathbf{W}=\hat{\mathbf{H}}\left(\hat{\mathbf{H}}^H\hat{\mathbf{H}}\right)^{-1}$, the SE of the mixed-ADC architecture with ZF detection becomes 
%\begin{equation}\label{SE formulaZF_mix}
%	\mathcal{S}_k^{\mathrm{ZF}}=\mathcal{R}\left(\frac{p\left(M-K\right)\sigma_{\hat{h}}^{2}}{pK\left(1-\sigma_{\hat{h}}^{2}\right)+\sigma_n^2+\frac{\left(1-\frac{2}{\pi}\right)}{\alpha^2}\left(1-\frac{N}{M}\right)}\right).
%\end{equation}
%ZFZFZFZFZFZFZFZFZFZFZZFFZFZFZZFZFZF

From (\ref{SE formulaMRC_mix}), it can be observed that the gain of exploiting the mixed-ADC architecture is manifested in the SE expressions by two factors, channel estimation improvement by a factor of $\sigma_{\hat{h}}^{2}$, and quantization noise reduction by a factor of $1-N/M$. 

%In particular, contrary to MRC detection where channel estimation improvement only enhances the gain of coherent combining of the received signal, it plays a pivotal role in interference rejection for ZF detection. Hence, we may find mixed-ADC architecture more useful when ZF detection is employed at the BS.  

\subsubsection{ Mixed-ADC\ Detection\ with Antenna\ Selection}\
 Having an accurate channel estimate can help us to employ the $N$ high-resolution ADCs in an intelligent manner to further improve the performance of the mixed-ADC architecture. A careful look at the SQINR expression in (\ref{uplink achievable rate}) reveals that the effect of one-bit quantization on the SE is manifested by the last term of the denominator. Hence, one can maximize the SE by minimizing this term through smart use of the $N$ high-resolution ADCs. We refer to this approach as \emph{Mixed-ADC with Antenna Selection}. We consider an antenna selection scheme suggested by the SQINR expression in (\ref{uplink achievable rate}).
In this approach, the $N$ high-resolution ADCs are connected to the antennas corresponding to rows of $\hat{\mathbf{H}}$ with the largest energy, i.e. $\sum_{k=1}^{K}{\left|\hat{h}_{mk}\right|^2}$.
Besides numerical evaluation in Section \ref{sec:Simulation}, in Theorem \ref{smart_mrc} we derive a bound for the SE achieved by MRC detection with antenna selection.
 \begin{theo}\label{smart_mrc}
 	The spectral efficiency of the mixed-ADC system with antenna selection and an MRC receiver is lower bounded by
 	\begin{equation}\label{MRC_smart}
 	\bar{\mathcal{S}}_k^{\mathrm{MRC}}=	\mathcal{R}\left(\frac{pM\sigma_{\hat{h}}^{2}}{pK+\sigma_n^2+\frac{\left(1-\frac{2}{\pi}\right)}{MK\alpha^2}{\left({\sum_{m=1}^{M-N}{\chi_m}}\right)}}\right),
 	\end{equation}
 	where $\chi_m$ is defined at the top of the next page, and $\mathcal{F}_A$ denotes the Lauricella function of type A \cite{Shi}.
 	\begin{figure*}[!t]
\begin{multline}\label{gamma order stat}
	\chi_m=
	\frac{M\,!}{\left(m-1\right)\,!\left(M-m\right)\,!}\sum_{\ell=0}^{M-m}
	(-1)^{-\ell}\binom {M-m}{\ell}\left(\Gamma(K)\right)^{-m-\ell}K^{1-m-\ell}\Gamma\left(1+K\left(m+\ell\right)\right)\\
	\qquad 
	\times \mathcal{F}_{A}^{(m+\ell-1)}\left(1+K\left(m+\ell\right);K,\cdots,K;K+1,\cdots,K+1;-1,\cdots,-1 \right)
\end{multline}
\hrulefill
    \vspace*{4pt}
    \end{figure*}
 \end{theo}
 \begin{proof}
See Appendix B.	
\end{proof}
% \begin{theo}\label{smart_zf}
% 	For $M\gg1$, the spectral efficiency of the mixed-ADC with antenna selection and ZF receiver is lower bounded by (\ref{ZF_smart}) shown at the top of the next page.
% 	\begin{figure*}[!t]
% 	\begin{equation}\label{ZF_smart}
% 	\bar{\mathcal{S}}_k^{\mathrm{ZF}}=	\mathcal{R}\left(\frac{p\left(M-K\right)\sigma_{\hat{h}}^{2}}{pK\left(1-\sigma_{\hat{h}}^{2}\right)+\sigma_n^2+\frac{\left(1-\frac{2}{\pi}\right)}{\alpha^2}\frac{\left(M-K\right)\left({\sum_{m=1}^{M-N}{\sum_{\ell=M-m+1}^{M}{\ell^{-1}}}}\right)}{{M^2}}}\right)
% 	\end{equation}
% 	\hrulefill
%    \vspace*{4pt}
%    \end{figure*}
% \end{theo}
% \begin{proof}
%See the Appendix.	
%\end{proof}
The lower bound (\ref{MRC_smart}) explicitly reflects the benefit of antenna selection in the data transmission phase. By comparing (\ref{MRC_smart}) with (\ref{SE formulaMRC_mix}), it is evident that antenna selection has improved the SE by replacing $1-N/M$ by $\frac{\sum_{m=1}^{M-N}{\chi_m}}{MK}$. In Section \ref{sec:Simulation} we illustrate how antenna selection improves SE for different SNRs. Note that Theorem \ref{smart_mrc} assumes the ability to make an arbitrary assignment of the high-resolution ADCs to different RF chains, which may not be possible if the ADC multiplexing is implemented by a bank of subarray switches. In the numerical results presented later, we show that this does not lead to a significant degradation in performance.\color{black}
%As an example, for typical case of $M=100$ and $N=20$, the former equals $0.8$ while the latter equals $0.48$ which can result in considerable improvement especially for high SNRs.   
%In the next Section, we  evaluate (\ref{SE formulaMRC_mix})-(\ref{ZF_smart}) for different channel estimation approaches.
\subsection{ ZF\ Detection}\label{sec:ZF_det}
In this section, we study the SE of the mixed-ADC architecture with ZF detection. 
%Since mixed-ADC SE expression with ZF detection is highly intractable, 
%For the case of random detection with perfect CSI, we derive a analytical expression up to an integral of a hypergeometric function.
 %For the random detection with imperfect CSI and the mixed-ADC with antenna selection, 
 %we clarify the preliminaries for numerical analysis.  
%\subsubsection{ Random\ Mixed-ADC\ Detection}\label{sec:Random_Mix}
To design a mixed-ADC adapted ZF detector, we re-write the last two terms of the denominator of (\ref{uplink achievable rate}) as follows:
\begin{equation}\label{zf_design}
	\boldsymbol{w}_k^H\left(\sigma_n^2\mathbf{I}_M+\alpha^{-2}\mathbf{C}_{\boldsymbol{q}_d}\right)\boldsymbol{w}_k=
	\biggl[\mathbf{W}^H\mathbf{C}_{n_{\mathrm{eff}}}\mathbf{W}\biggr]_{kk},
\end{equation} 
where $\mathbf{C}_{n_{\mathrm{eff}}}=\sigma_n^2\mathbf{I}_M+\alpha^{-2}\mathbf{C}_{\boldsymbol{q}_d}$. Accordingly, the ZF detector for the mixed-ADC architecture can be written as
\begin{equation}\label{ZF_detector_mixed}
	\mathbf{W}=\mathbf{C}_{n_{\mathrm{eff}}}^{-1}\hat{\mathbf{H}}\left(\hat{\mathbf{H}}^H\mathbf{C}_{n_{\mathrm{eff}}}^{-1}\hat{\mathbf{H}}\right)^{-1}.
\end{equation}
Plugging (\ref{ZF_detector_mixed}) into (\ref{uplink achievable rate}) yields (\ref{ZF_SQINR}) at the top of the next page. 
\begin{figure*}[!t]
\begin{equation}\label{ZF_SQINR}
%\resizebox{.9\hsize}{!}
\mathrm{SQINR}_k^{\mathrm{ZF}}=
\frac{p}{pK\left(1-\sigma_{\hat{h}}^{2}\right)\mathbb{E}\left\{\biggl[\left(\hat{\mathbf{H}}^H\mathbf{C}_{n_{\mathrm{eff}}}^{-1}\hat{\mathbf{H}}\right)^{-1}
\hat{\mathbf{H}}^H\mathbf{C}_{n_{\mathrm{eff}}}^{-2}\hat{\mathbf{H}}
\left(\hat{\mathbf{H}}^H\mathbf{C}_{n_{\mathrm{eff}}}^{-1}\hat{\mathbf{H}}\right)^{-1}\biggr]_{kk}\right\}
+
 \mathbb{E}\left\{\biggl[\left(\hat{\mathbf{H}}^H\mathbf{C}_{n_{\mathrm{eff}}}^{-1}\hat{\mathbf{H}}\right)^{-1}\biggr]_{kk}\right\}}
\end{equation}
\hrulefill
    \vspace*{4pt}
    \end{figure*}
Similar to the MRC case, the SQINR in (\ref{ZF_SQINR}) suggests the same antenna selection approach for ZF detection. In general, calculating the expected values in (\ref{ZF_SQINR}) is not tractable neither for arbitrary-antenna mixed-ADC detection nor mixed-ADC with antenna selection. Hence, we numerically evaluate the performance of mixed-ADC with ZF detection in the next section.

\subsection{ Massive\ MIMO\ with\ Uniform\ ADC\ Resolution}\label{sec:pure low}
Contrary to the mixed-ADC architecture where the ADC comparators are concentrated in a few antennas, uniformly spreading the comparators over the array is an alternative approach \cite{Sarajlik,Verenzuela,Fan,Li Fan,Qiao}. In this subsection, we provide the SE expressions for such systems. These expressions will be used in the next section for performance comparisons with the mixed-ADC architecture.

The SE for the case of all one-bit ADCs was derived in \cite{Li} using the Bussgang decomposition. For ADC resolutions of 2 bits or higher, the AQNM model is sufficiently accurate.  Using AQNM and following the same reasoning as in \cite{Fan,Li Fan,Qiao}, the SE of a massive MIMO system with uniform resolution ADCs can be derived as
\begin{multline}\label{SE_pure_MRC}
	\tilde{\mathcal{S}}_k^{\mathrm{MRC}}=\\
	\mathcal{R}\left(\frac{pM\tilde{\sigma}_{\hat{h}}^{2}}{pK+\sigma_n^2+\frac{\left(1-\alpha_{0}\right)}{\alpha_{0}^2}\left(p\left(\tilde{\sigma}_{\hat{h}}^{2}+K\right)+\sigma_n^2\right)}\right)
\end{multline}
\begin{multline}\label{SE_pure_ZF}
	\tilde{\mathcal{S}}_k^{\mathrm{ZF}}=\\
	\mathcal{R}\left(\frac{p\left(M-K\right)\tilde{\sigma}_{\hat{h}}^{2}}{pK\left(1-\tilde{\sigma}_{\hat{h}}^{2}\right)+\sigma_n^2+\frac{\left(M-K\right)\tilde{\sigma}_{\hat{h}}^{2}}{\alpha^2}\mathbb{E}\left\{{{\boldsymbol{w}}}_{k}^{H}\mathbf{C}_{0}{\boldsymbol{w}}_{k}\right\}}\right),
\end{multline}
for MRC and ZF detection, respectively. In (\ref{SE_pure_MRC}) and (\ref{SE_pure_ZF}),
\begin{equation}\label{AQNM_sigma} 
\tilde{\sigma}_{\hat{h}}^{2}=\frac{\alpha_{0}^2\eta p}{\alpha_{0}^2\eta p+\alpha_{0}^2\sigma_n^2+\alpha_{0}\left(1-\alpha_{0}\right)\left(pK+\sigma_n^2\right)},
\end{equation}
 $\alpha_0$ is a scalar depending on the ADC resolution and can be found in Table I of \cite{Fan}, $\boldsymbol{w}_{k}$ is the $k$th column of $\mathbf{W}=\hat{\mathbf{H}}\left(\hat{\mathbf{H}}^H\hat{\mathbf{H}}\right)^{-1}$, and $\mathbf{C}_{0}$ denotes the covariance matrix of the quantization noise based on the AQNM model \cite{Fan}. The detailed calculation of $\mathbb{E}\left\{{{\boldsymbol{w}}}_{k}^{H}\mathbf{C}_{0}{\boldsymbol{w}}_{k}\right\}$ in (\ref{SE_pure_ZF}) is provided in \cite{Qiao} which we do not include here for the sake of brevity.
\color{black}
\section{Numerical Results}\label{sec:Simulation}
By substituting from (\ref{ghat variance}), (\ref{ghat variance2}), and (\ref{ghat variance3}) into (\ref{SE formulaMRC_mix}), (\ref{MRC_smart}), and (\ref{ZF_SQINR}), we can evaluate the performance of mixed-ADC architectures for different system settings. For all of the following experiments, \color{black}we consider a system with  $M=100$ antennas at the BS, and $K=10$ users. Also, we assume the power control approach of \cite{Emil Bj} is used, so that $p_k \beta_k =p$ for all $k$. We also assume that an optimal resource allocation has been performed \cite{Li Fan,Ngo2} such that the training length, $\eta_{\mathrm{eff}}$, transmission power during the training phase, $p_t$, and data transmission phase, $p_d$ are optimized under a power constraint $\eta_{\mathrm{eff}}p_t+(T-\eta_{\mathrm{eff}})p_d=P_{\mathrm{ave}}T$. In the following figures, the SNR is defined as $\mathrm{SNR}\triangleq P_{\mathrm{ave}}/\sigma_n^2$.
\color{black} 
\begin{figure}
\centering
\includegraphics[width=0.5\textwidth]
{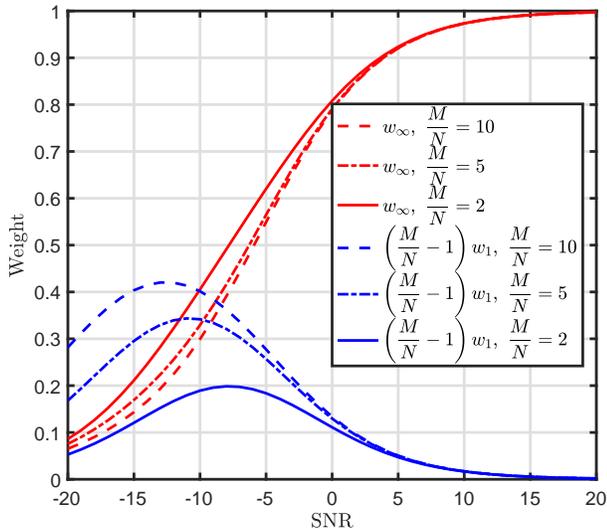}
\caption{Weights used in the LMMSE channel estimator for high-resolution and one-bit observations.}
\label{fig2}
\end{figure}

Fig. \ref{fig2}  illustrates the optimal weights for combining high-resolution and one-bit observations for the joint high-resolution/one-bit LMMSE channel estimation. Interestingly, it can be seen that when $M/N$ is large, the one-bit observations are emphasized in the low SNR regime relative to the high-resolution observations. In addition, in contrast to the weights for the high-resolution observations, which rise monotonically with increasing SNR, the weight for the one-bit observations grows at first and then decreases to zero. 
\begin{figure}[t!]
\centering
\includegraphics[width=0.5\textwidth]
{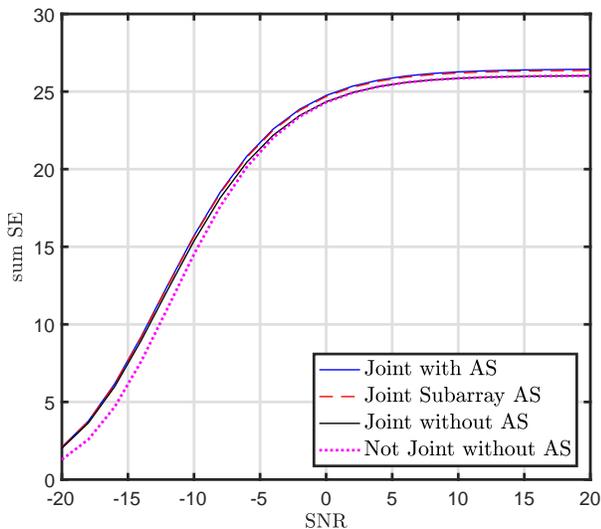}
\caption{Sum SE for MRC detection versus SNR for $M=100$, $N=20$, and $T=400$.}
\label{fig5}
\end{figure}
\begin{figure}[t!]
\centering
\includegraphics[width=0.5\textwidth]
{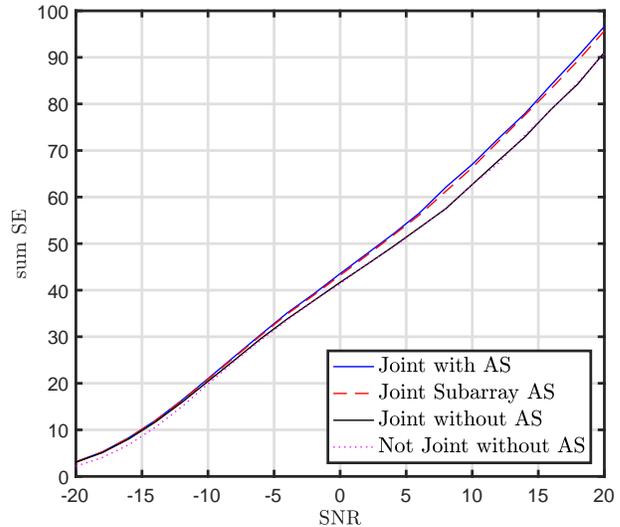}
\caption{Sum SE for ZF detection versus SNR for $M=100$, $N=20$, and $T=400$.}
\label{fig6}
\end{figure}

To study the performance improvement due to joint channel estimation and antenna selection in mixed-ADC massive MIMO, the sum SE for the MRC and ZF detectors for a system with coherence interval $T=400$ symbols and $N=20$ high-resolution ADCs is depicted in Fig. \ref{fig5} and Fig. \ref{fig6}, respectively. In these and subsequent figures, ``Joint with AS'' indicates that the channel estimation was performed with both one-bit and high-resolution ADCs and that antenna selection (AS) was used for data detection, ``Joint without AS'' represents the same case without antenna selection, ``Joint Subarray AS'' means that the antenna selection only occurred within each $M/N$-element subarray (one high-resolution ADC assigned to the strongest channel within each subarray), and ``Not Joint without AS'' represents the case in which channel is estimated based on only high-resolution observations and no antenna selection is employed\color{black}. 
%``One-bit'' refers to the performance obtained when only $M$ one-bit antennas are used for both channel estimation and data detection, ``Full resolution'' indicates that the channel estimation was performed using only the full-resolution ADCs while both full- and one-bit resolution ADCs were used for data detection, ``Joint'' indicates that the channel estimation and data detection used both the full-resolution and one-bit ADCs, and ``AS'' indicates when antenna selection was used for the data decoding.
%%%For MRC detection, when $N=20$, the mixed-ADC architecture can improve the sum SE for all SNRs, but when $N=10$ an all-one-bit architecture is better due to the larger training overhead incurred when $N$ is smaller. 
%%%When $N=20$, the benefit of the mixed-ADC architecture is reduced at high SNR due to the diminished ability of the one-bit ADCs to compensate for the reduced data transmission period, and the inability of MRC detection to reject the interference for high SNRs. 
For both MRC and ZF, it can be seen that antenna selection slightly improves the SE for high SNRs, where the channel estimation is most accurate. At low SNR, we see that joint channel estimation provides a gain from the use of one-bit ADCs, which provide useful information at these SNRs. We also see that the constrained AS required when the switching is only performed within subarrays provides nearly identical performance to the case where arbitrary AS is allowed.
%while for low SNRs joint channel estimation provides the main contribution to SE improvement. The former is due to the high channel estimation accuracy for high SNRs. Moreover, since the performance degradation of one-bit ADCs is severe for high SNRs, smart utilization of high-resolution ADCs is more influential in this regime. The latter is due to the fact that at low SNR, the joint estimation benefits from the use of one-bit ADCs, which provide useful information. A similar argument can be made for ZF detection.

 Note that the main reason for the small gain for antenna selection is due to the fact that, with multiple users, selecting a given antenna does not benefit all users simultaneously, and the strong users responsible for a given antenna being selected will in general be different for different antennas. Thus, the improvement due to increased signal-to-noise ratio for some users is somewhat offset by the fact that other users may experience a lower SNR on those same antennas. We would see a much larger benefit for antenna selection if only a single user were present. 
%Although the mixed-ADC architecture is beneficial when the number of high-resolution ADCs is sufficiently large, it fails to significantly increase the SE except for very low SNRs and just for the case with joint channel estimation.
\color{black}

Figs. \ref{fig7} and \ref{fig8} provide a comparison among a mixed-ADC massive MIMO system with joint channel estimation and antenna selection, an all-one-bit architecture (``One-bit''), and a mixed-ADC without round-robin training for which the high-resolution ADCs are connected to a fixed set of antennas without ADC switching or antenna selection (``Non-round-robin'') \cite{Hessam2}. 
%When $N=20$, the benefit of the mixed-ADC architecture is reduced at high SNR due to the diminished ability of the one-bit ADCs to compensate for the reduced data transmission period, and the inability of MRC detection to reject the interference for high SNRs.
Since mixed-ADC channel estimation improves the channel estimation accuracy by expending a larger portion of the coherence interval for training, its benefit is directly related to the length of the coherence interval. For MRC detection, when $T=400$, the mixed-ADC architecture performs better than the all-one-bit architecture for $N=20$, but when $N=10$ the all-one-bit architecture is better due to the larger training overhead incurred when $N$ is smaller.  
However, for $T=1000$, mixed-ADC outperforms the all-one-bit architecture at high SNRs for both $N=10,~20$, while the all-one-bit case is still better for $N=10$ at low SNRs. Round-robin training provides better SE performance at high SNR when $N=20$ compared to the case without antenna switching, especially for the larger coherence interval. However, for other cases, the round-robin training overhead significantly reduces the SE, especially for $N=10$ and the shorter coherence interval. \color{black}

%Hence, it fails to significantly increase the SE except for very low SNRs and just for the case with joint channel estimation.  

For ZF detection, we see that the mixed-ADC architectures can provide very large gains in SE compared to the one-bit case at high SNRs, regardless of $T$. For low SNRs, there is little to no improvement. These cases still do not show a significant benefit for round-robin training compared with a fixed ADC assignment; only when $N=20$ and $T=1000$ do we see a slight improvement. 
%On the other hand, for ZF detection with $T=400$ and $N=20$, the sum SE improvement for the mixed-ADC architecture is noticeable for high SNRs. Interestingly, we observe that the mixed-ADC architecture can overcome the high SNR saturation of the all-one-bit architecture and considerably increases the sum SE. Similar to MRC detection, when $N=10$ the mixed-ADC architecture cannot compensate for the penalty due to the increased training duration. Hence, for moderate SNRs, an all-one-bit architecture performs slightly better than the mixed-ADC architecture when $N \le 10$. However, when $T=1000$, even only $N=10$ high-resolution ADCs can lead to a performance improvement for moderate SNRs while the benefit of the mixed-ADC architecture is more pronounced for $N=20$.  
\begin{figure}[t!]
\centering
\includegraphics[width=0.5\textwidth]
{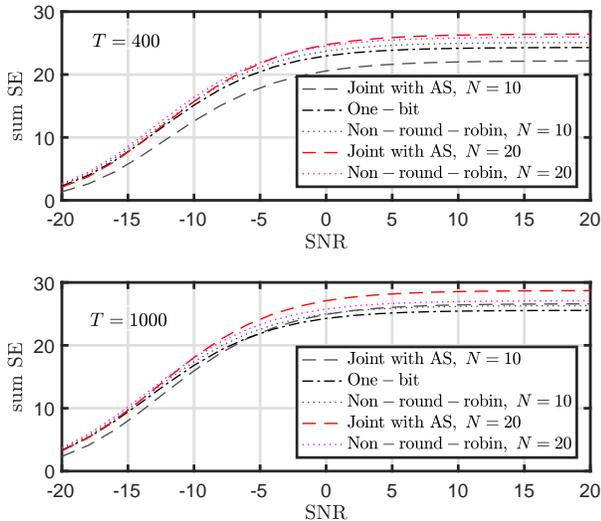}
\caption{Sum SE for MRC detection versus SNR for $M=100$, $N=20,~10$, and $T=400,~1000$.}
\label{fig7}
\end{figure}

\begin{figure}[t!]
\centering
\includegraphics[width=0.5\textwidth]
{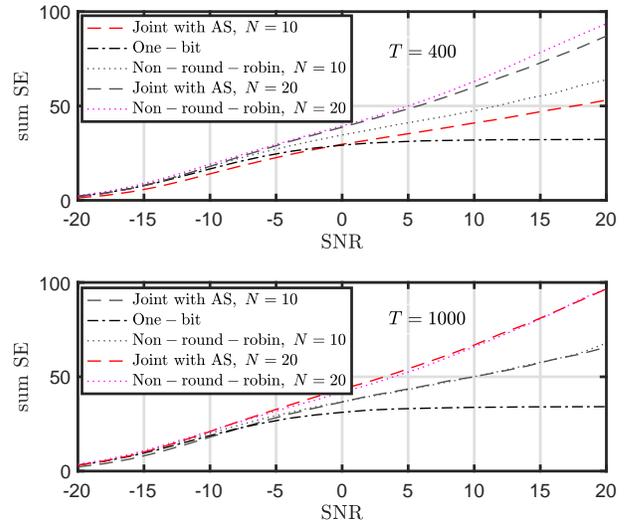}
\caption{Sum SE for ZF detection versus SNR for $M=100$, $N=20,~10$, and $T=400,~1000$.}
\label{fig8}
\end{figure}
%Since mixed-ADC channel estimation improves the channel estimation accuracy by expending a larger portion of the coherence interval for training, its benefit is directly related to the length of the coherence interval. Fig. \ref{fig5} and Fig. \ref{fig6} show the sum SE for a system with $T=400$ and $N=20,~10$. 
%For the MRC receiver, the improvement when $N=20$ is evident. For $N=10$, mixed-ADC is beneficial for a larger range of SNRs but still performs worse than an all-one-bit architecture for very large SNRs. 

For $N=20$, Figs. \ref{fig9} and \ref{fig10} show how the coherence interval $T$ impacts the effectiveness of the mixed-ADC architecture for MRC and ZF detectors, respectively. 
%In these figures we have included the performance of mixed-ADC without round-robin training, for which the high-resolution ADCs are connected to a fixed set of antennas without ADC switching \cite{Hessam2}. 
For mixed-ADC MRC detection, it is apparent that the best  choice among the three architectures (all one-bit, mixed-ADC with and without round-robin training) depends on the SNR operating point and the length of the coherence interval. The advantage of round-robin training becomes apparent for long coherence intervals, where the increased training length has a smaller impact. The gain for round-robin training is greatest at higher SNRs. For shorter coherence intervals, mixed ADC with fixed antenna/ADC assignments provides the best SE, with the largest gains again coming at higher SNRs. For this value of $N$, the all-one-bit system generally has the lowest SE, although the difference is not large for MRC. 

%For example, for very low SNRs, the mixed-ADC architecture without round-robin training has the best performance. However, for high SNRs and sufficiently large coherence interval, the mixed-ADC architecture with round-robin training is the best choice. The same observation can be made for mixed-ADC ZF detection. Considering the increased power consumption incurred by the mixed-ADC architecture compared with the all-one-bit system, these observations give insight into the compromise between SE and power consumption for different SNRs, different detection methods, and channel characteristics.  

\begin{figure}[t!]
\centering
\includegraphics[width=0.5\textwidth]
{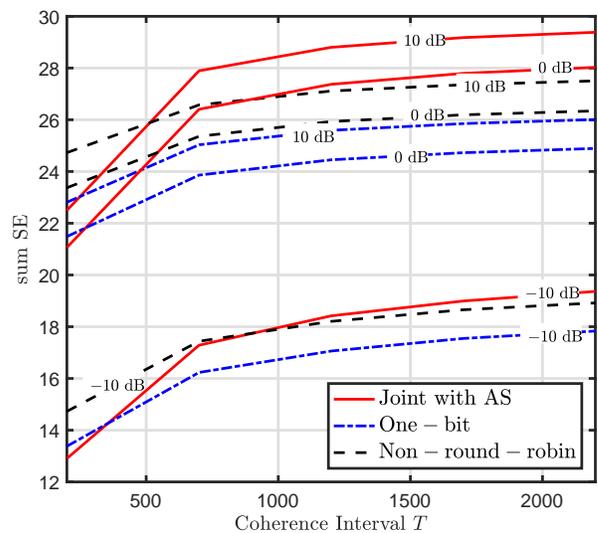}
\caption{Sum SE for MRC detection versus $T$ for $M=100$, $N=20$, and $SNR=-10, 0, 10$ dB.}
\label{fig9}
\end{figure}

\begin{figure}[t!]
\centering
\includegraphics[width=0.5\textwidth]
{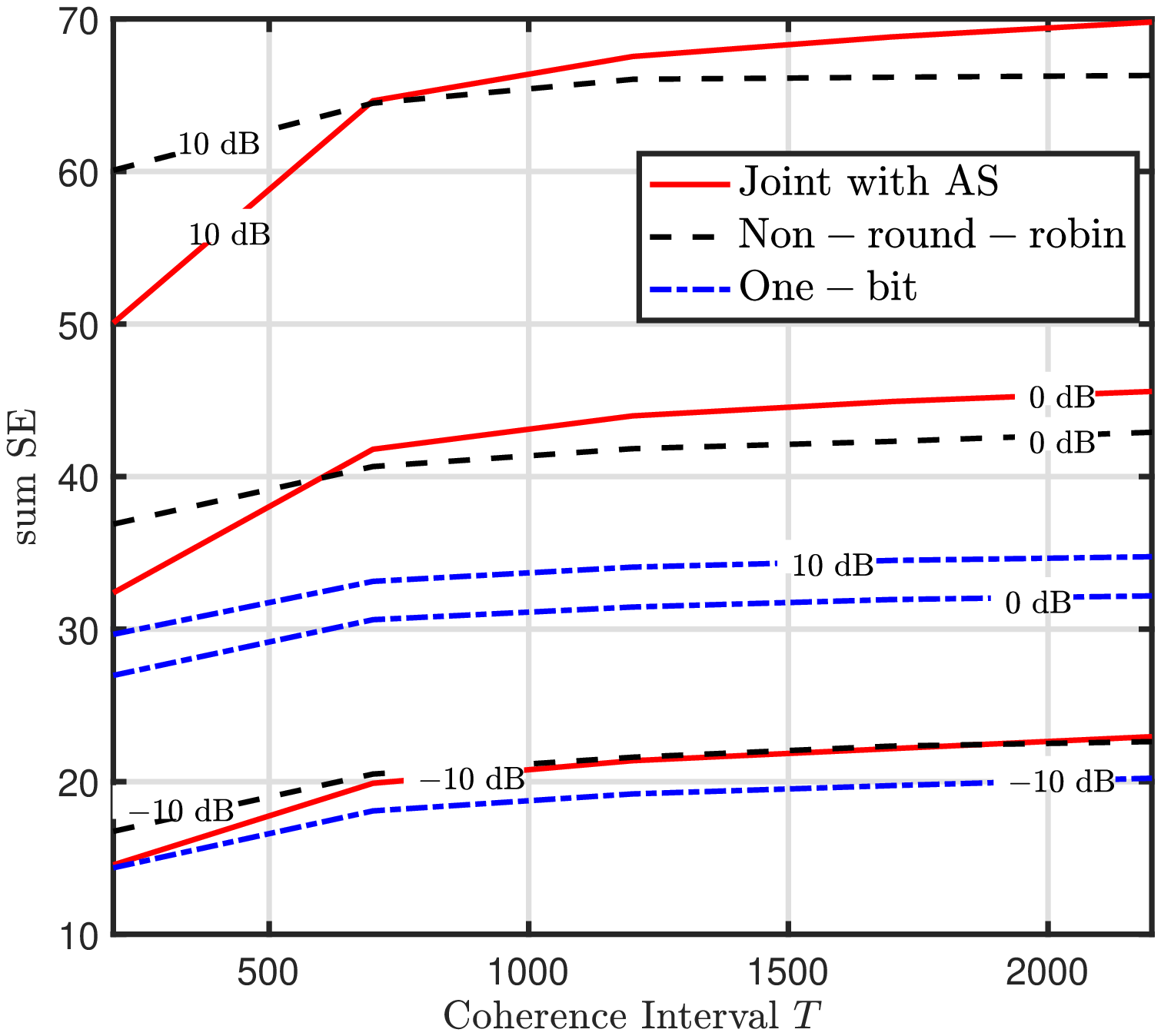}
\caption{Sum SE for ZF detection versus $T$ for $M=100$, $N=20$, and $SNR=-10, 0, 10$ dB.}
\label{fig10}
\end{figure}

%Another important comparison is to investigate the benefit of using the extra one-bit ADCs beyond the $N$ high-resolution ADCs. Figs. \ref{fig9} and \ref{fig10} respectively illustrate the gain of using a large number of one-bit ADCs beyond the $N=20$ high-resolution ADCs for MRC and ZF. 
%For MRC detection, the performance gain is significant for all SNRs. In other words, although adding more one-bit ADCs requires a longer training period, the improvement due to coherent combining by the MRC detector not only compensates for the channel estimation overhead, but also considerably improves the SE for all SNRs. This is not the case for ZF detection. For low to moderate SNRs, the addition of the one-bit ADCs outperforms a system with a small number of high-resolution ADCs. But as the SNR increases, the one-bit ADCs become less useful. Ultimately, the performance of mixed-ADC with ZF falls below that achieved with only $N=20$ high-resolution antennas due to the channel estimation overhead. This suggests that one should disable the one-bit ADCs in the high SNR regime for ZF detection. 
%
%
\begin{figure}[t!]
\centering
\includegraphics[width=0.5\textwidth]
{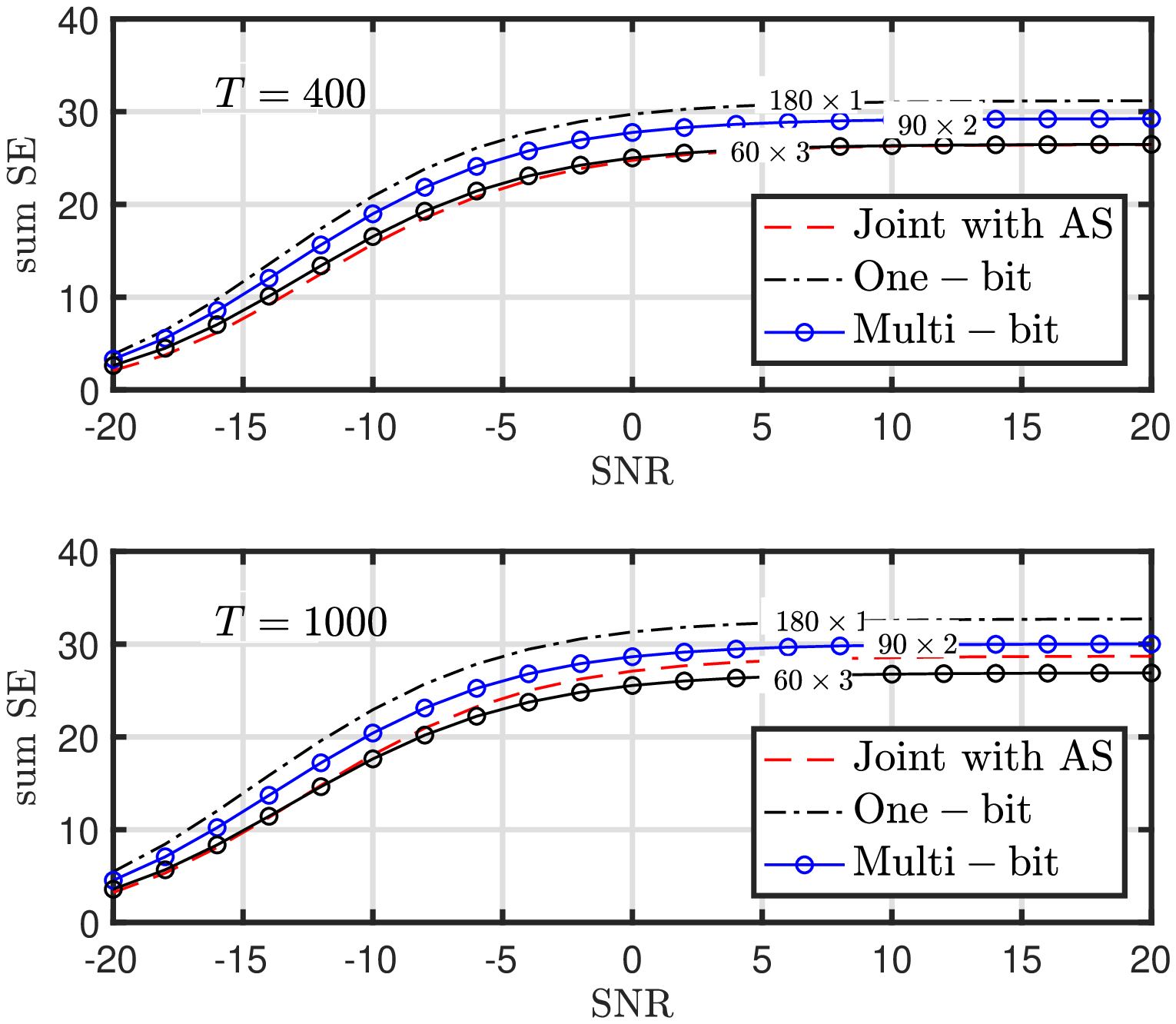}
\caption{Sum SE for MRC detection versus SNR for $180$ comparators and $T=400,~1000$.}
\label{fig7b}
\end{figure}

\begin{figure}[t!]
\centering
\includegraphics[width=0.5\textwidth]
{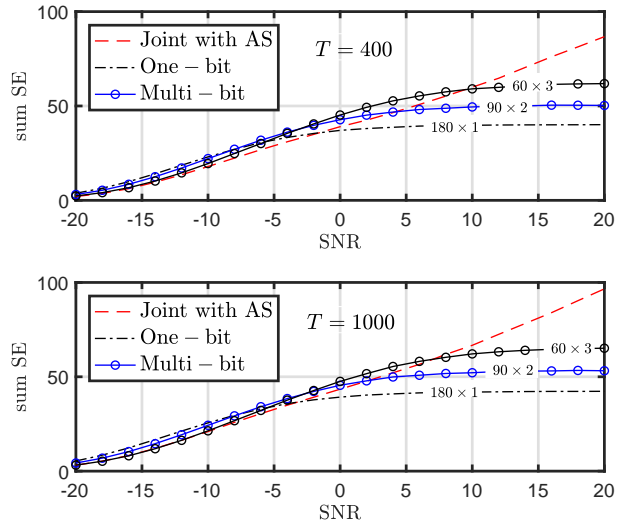}
\caption{Sum SE for ZF detection versus SNR for $180$ comparators and $T=400,~1000$.}
\label{fig8b}
\end{figure}

The next example investigates the impact of distributing the resolution (i.e., the comparators of the ADCs) across the array with different numbers of antennas. If we assume that the ``high-resolution'' ADCs consist of 5 bits \cite{Roth}, a mixed-ADC architecture with $N=20$ high-resolution and $M-N=80$ one-bit ADCs will have 180 total comparators. Figs.~ \ref{fig7b} and \ref{fig8b} illustrate the SE achieved by distributing the 180 comparators across arrays of different length for MRC and ZF detection, respectively. In these figures, ``Joint with AS'' and ``Non-round-robin'' refer to mixed-ADC architectures with $N=20$ 5-bit ADCs and $M-N=80$ one-bit ADCs, ``One-bit'' corresponds to $M=180$ antennas with one-bit ADCs, and ``Multi-bit'' indicates a system with either $M=90$ 2-bit ADCs or $M=60$ 3-bit ADCs. As we see in the figures, it can be inferred that for MRC detection, which is interference limited, it is better to have a larger number of antennas with lower-resolution ADCs instead of equipping the BS with fewer antennas and high resolution ADCs. This is consistent with the results of \cite{Park,Hessam}, and is due to the fact that a larger number of antennas helps the system to more effectively cancel the interference. On the other hand, for ZF detection which is noise limited, the use of high-resolution ADCs avoids additional quantization noise imposed by the low-resolution ADCs, and is more beneficial than having a larger number of antennas with low-resolution ADCs at high SNR.        

Finally, Figs. \ref{fig11} and \ref{fig12} show the impact of the number of high-resolution ADCs in a mixed-ADC system with $M=100$ antennas, $K=10$ users, and various numbers $N$ of high-resolution ADCs, where $N=100$ denotes the all-high-resolution system. It is apparent that with a large enough coherence interval and a sufficient number of high-resolution ADCs, the mixed-ADC implementation with joint round-robin channel estimation and antenna selection outperforms the all-one-bit architecture and mixed-ADC without round-robin training. The gains are greatest when ZF detection is used and the SNR is high, but such gains must be weighed against the increased power consumption and hardware complexity.\color{black}

\begin{figure}[t!]
\centering
\includegraphics[width=0.5\textwidth]
{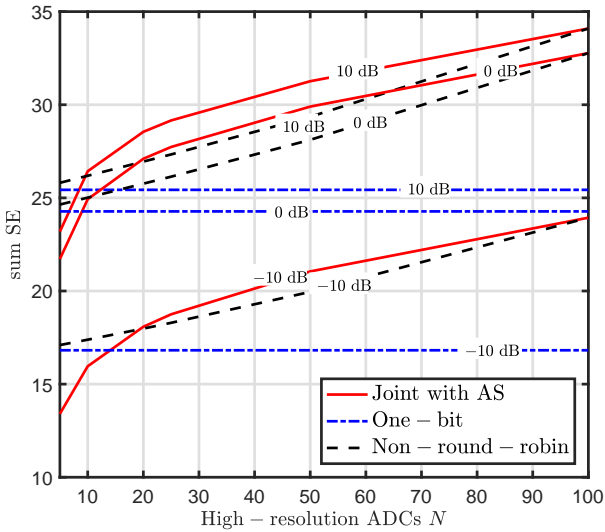}
\caption{Sum SE for MRC detection versus $N$ for $SNR=-10, 0, 10$ dB and $T=1000$.}
\label{fig11}
\end{figure}

\begin{figure}[t!]
\centering
\includegraphics[width=0.5\textwidth]
{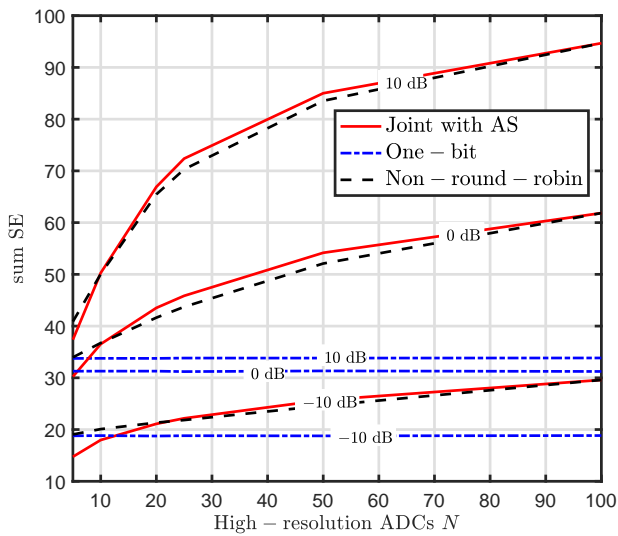}
\caption{Sum SE for ZF detection versus $N$ for $SNR=-10, 0, 10$ dB and $T=1000$.}
\label{fig12}
\end{figure}

\section{Conclusion}\label{sec:concolusion}
%\vspace{-1mm}
We studied the spectral efficiency of mixed-ADC massive MIMO systems with either MRC or ZF detection. We showed that properly accounting for the impact of the quantized receivers using the Bussgang decomposition is important for obtaining an accurate analysis of the SE. We introduced a joint channel estimation approach to leverage both high-resolution ADCs and one-bit ADCs and our analytical and numerical results confirmed the benefit of joint channel estimation for low SNRs. 

Mixed-ADC detection with MRC and ZF detectors and antenna selection were also studied. Analytical expressions were derived for MRC detection and a numerical performance analysis was performed for ZF detection. It was shown that antenna selection provides a slight advantage for high SNRs while this advantage tends to disappear for low SNRs.  

We showed that the SNR, the number of high-resolution ADCs and the length of the coherence interval play a pivotal role in determining the performance of mixed-ADC systems. We showed that, in general, mixed-ADC architectures will have the greatest benefit compared to implementations with all low-resolution ADCs when ZF detection is used and the SNR is relatively high. In such cases, the gain of the mixed-ADC approach can be substantial.  Gains are also possible for MRC, but they are not as significant, and require larger numbers of high-resolution ADCs to see a benefit compared with the ZF case. The more complicated mixed-ADC approach based on ADC switching and round-robin training can achieve the best performance in some cases, particularly when the coherence interval is long and more high-resolution ADCs are available to reduce the number of training interval repetitions. Otherwise, a mixed-ADC implementation without ADC switching and extra training is preferred.\color{black}  
%MRC detection requires long coherence intervals and a sufficient number of high-resolution ADCs to outperform an all-one-bit architecture. On the other hand, mixed-ADC with ZF detection outperforms the all-one-bit system even for coherence intervals of moderate length. Moreover, mixed-ADC with ZF detection can overcome the high SNR saturation that is inevitable for all-one-bit architectures.
 %Large coherence interval duration, allows for utilizing less high-resolution ADCs and still achieving a better SE than a system with all-one-bit architecture. 
 %However, the energy consumption imposed by adding high-resolution ADCs and the resulting SE gain should be considered for designing a mixed-ADC massive MIMO system.   

%performance of one-bit massive MIMO systems when a mixed-ADC architecture is used for channel estimation. It was shown that by equipping the BS with $N\ll M$ high-resolution ADCs and expending a larger portion of the coherence interval for training, one can improve the SE of one-bit massive MIMO systems while imposing a tolerable power consumption burden on the BS.  
\color{black} 
\section*{Appendix}\label{sec:appendix}
\subsection{Proof of Theorem \ref{theo3} }
From (\ref{training matrix}), the observations from the high-resolution ADCs can be written as
\begin{equation}\label{full_res_obser}
	\boldsymbol{v}(0)=\sqrt{\frac{1}{\eta p_k}}\mathbf{X}\boldsymbol{\phi_k^{*}}=\boldsymbol{g}_k+\tilde{\boldsymbol{n}}(0),
\end{equation}
where $\tilde{\boldsymbol{n}}(0)\sim\mathcal{CN}(\mathbf{0},\frac{\sigma_n^2}{\eta p_k}\boldsymbol{I}_M)$. In addition, from (\ref{matrix quantized training}), the observations from the one-bit ADCs become
\begin{equation}\label{one-bit_observ}
	\boldsymbol{v}(t)=\sqrt{\frac{1}{\eta p_k}}\mathbf{Y}_t\boldsymbol{\bar{\phi_k}^{*}}=\boldsymbol{g}_k+\tilde{\boldsymbol{n}}(t)+\tilde{\boldsymbol{q}}(t),~~t\in\mathcal{T},
\end{equation} 
where $\tilde{\boldsymbol{n}}(t)\sim\mathcal{CN}(\mathbf{0},\frac{\sigma_n^2}{\eta p_k}\boldsymbol{I}_M)$ is independent of $\tilde{\boldsymbol{n}}(t')$ for $t\neq t'$, and $\tilde{\boldsymbol{q}}(t)=\sqrt{\frac{1}{\eta p_k}}\mathbf{Q}(t)\boldsymbol{\bar{\phi_k}^{*}}$. Since the elements of $\boldsymbol{v}(t)$ are independent, we can estimate the $m$th channel $g_{mk}$ separately. Therefore, stacking all the observations in a vector, we can write
\begin{equation}\label{stacked_obser}
	\underbrace{\begin{bmatrix}
        v_m(0) \\
        \vdots \\
        v_m(t) \\
        \vdots \\
        v_m(\frac{M}{N}-1)
\end{bmatrix}}_{\mathbf{v}}=\underbrace{\begin{bmatrix}
        1 \\
        \vdots \\
        1 \\
        \vdots \\
        1
\end{bmatrix}}_{\mathbf{1}_{\frac{M}{N}}}g_{mk}+
\underbrace{\begin{bmatrix}
        \tilde{n}_m(0) \\
        \vdots \\
        \tilde{n}_m(t)+\tilde{q}_m(t) \\
        \vdots \\
        \tilde{n}_m(\frac{M}{N}-1)+\tilde{q}_m(\frac{M}{N}-1)
\end{bmatrix}}_{\mathbf{u}}.
\end{equation}

As a result, the LMMSE estimation of the $m$th channel coefficient for the $k$th user is \cite{Kay}
\begin{equation}\label{estimate1_appen}
\hat{g}_{mk}=\left(\frac{1}{\beta_k}+\mathbf{1}_{\frac{M}{N}}^T\mathbf{C}_{\mathbf{u}}^{-1}\mathbf{1}_{\frac{M}{N}}\right)^{-1}\mathbf{1}_{\frac{M}{N}}^{T}\mathbf{C}_{\mathbf{u}}^{-1}\mathbf{v}.
\end{equation}
In Eq. (\ref{estimate1_appen}), $\mathbf{C}_{\mathbf{u}}$ denotes the covariance matrix of $\mathbf{u}$ which is a block diagonal matrix of the form
\begin{equation}\label{block_diagonal}
	\mathbf{C}_{\mathbf{u}}=
	\begin{bmatrix}
    \frac{\sigma_n^2}{\eta p_k} & 0 & \dots  & 0 \\
    0 & \sigma^2_{w_k} & \dots  & \varrho_k \\
    \vdots & \vdots & \ddots & \vdots \\
    0 & \varrho_k & \dots  & \sigma^2_{w_k}
\end{bmatrix}=
\begin{bmatrix}
    \frac{\sigma_n^2}{\eta p_k} & \mathbf{0}  \\
    \mathbf{0} & \mathbf{S}  \\
\end{bmatrix},
\end{equation}
where 
\begin{equation}\label{var_rho_def}
	\varrho_k=\mathbb{E}\{\left(\tilde{n}_m(t)+\tilde{q}_m(t)\right)\left(\tilde{n}_m(t')+\tilde{q}_m(t')\right)^{*}\},~~t\neq t',
\end{equation}
can be easily calculated with the aid of the Bussgang decomposition and the arcsine law as in (\ref{covariance}). 
Substituting (\ref{block_diagonal}) into (\ref{estimate1_appen}), we have
\begin{multline}\label{estimate2_appen}
\hat{g}_{mk}=\left(\frac{1}{\beta_k}+\frac{\eta p_k}{\sigma_n^2}+\mathbf{1}_{\frac{M}{N}-1}^T\mathbf{S}^{-1}\mathbf{1}_{\frac{M}{N}-1}\right)^{-1}\\
\times\Bigg[\frac{\sigma_n^2}{\eta p_k}~~\mathbf{1}_{\frac{M}{N}-1}^{T}\mathbf{S}^{-1}\Bigg]\mathbf{v}.
\end{multline}
To calculate the inverse of the matrix $\mathbf{S}$, we re-write it as 
\begin{equation}\label{S_form}
	\mathbf{S}=\left(\sigma^2_{w_k}-\varrho_k\right)\mathbf{I}_{\frac{M}{N}-1}+\varrho_k\mathbf{1}_{\frac{M}{N}-1}\mathbf{1}_{\frac{M}{N}-1}^T,
\end{equation}
and use Woodbury's matrix identity:
\begin{multline}\label{woodbury}
\mathbf{S}^{-1}=\frac{1}{\sigma_{w_k}^{2}-\varrho_k}\mathbf{I}_{\frac{M}{N}-1}-\\
\frac{1}{\left(\sigma_{w_k}^{2}-\varrho_k\right)^2}\left(\frac{1}{\varrho_k}+\frac{(\frac{M}{N}-1)}{\sigma_{w_k}^{2}-\varrho_k}\right)^{-1}\mathbf{1}_{\frac{M}{N}-1}\mathbf{1}_{\frac{M}{N}-1}^T,
\end{multline}
which yields
\begin{equation}\label{1TS-1}
	\mathbf{1}_{\frac{M}{N}-1}^T\mathbf{S}^{-1}=\frac{1}{\sigma^2_{w_k}+\left(\frac{M}{N}-2\right)\varrho_k}\mathbf{1}_{\frac{M}{N}-1}^T,
\end{equation}
\begin{equation}\label{1TS-11}
	\mathbf{1}_{\frac{M}{N}-1}^T\mathbf{S}^{-1}\mathbf{1}_{\frac{M}{N}-1}=\frac{\left(\frac{M}{N}-1\right)}{\sigma^2_{w_k}+\left(\frac{M}{N}-2\right)\varrho_k}.
\end{equation}
Substituting (\ref{1TS-1}) and (\ref{1TS-11}) into (\ref{estimate2_appen}) completes the proof. 

\subsection{Proof of Theorem \ref{smart_mrc} }\label{B}
Denote the energy of the $m$th row, $m\in\mathcal{M}=\{1,...,M\}$, of $\hat{\mathbf{H}}$ by $\mathcal{E}_{m}$, i.e.,
\begin{equation}\label{energy}
	\mathcal{E}_{m}\triangleq\sum_{k=1}^{K}{\left|\hat{h}_{mk}\right|^2}.
\end{equation}
To do antenna selection, we must connect the $N$ high-resolution ADCs to the antennas corresponding to the largest $\mathcal{E}_m$. Suppose that the indices of the $N$ antennas to which the high-resolution ADCs are connected are contained in the set $\mathcal{N}$. Hence, we have
\begin{multline}\label{proof_smart_MRC_1}
\sum_{k=1}^{K}{\mathbb{E}\left\{{\hat{\boldsymbol{h}}}_{k}^{H}\mathbf{C}_{\boldsymbol{q}_d}\hat{\boldsymbol{h}}_{k}\right\}}=\\
K\mathbb{E}\left\{{\hat{\boldsymbol{h}}}_{k}^{H}\mathbf{C}_{\boldsymbol{q}_d}\hat{\boldsymbol{h}}_{k}\right\}=
\left(1-\frac{2}{\pi}\right)\sum_{\mathcal{M}\backslash\mathcal{N}}^{}{\mathbb{E}\{\mathcal{E}_m\}}.	
\end{multline}
Eq. (\ref{proof_smart_MRC_1}) provides a criterion for connecting the $N$ high-resolution ADCs in the data transmission phase. In fact, it states that, for the MRC receiver, the expected value in (\ref{proof_smart_MRC_1}) will be minimized if the high-resolution ADCs are connected to the antennas corresponding to the largest $\mathcal{E}_m$. Denote $\mathcal{E}_{(m)}$ as the $m$th smallest value of $\mathcal{E}_m$, i.e.,
\[
\mathcal{E}_{(1)}\leq\mathcal{E}_{(2)}\leq\cdots\leq\mathcal{E}_{(M)}.
%\footnote{Throughout the proof the subscripts with parenthesis denote ordered random variable while subscript without parenthesis denote unordered random variable.}
\]
%Therefore, 
%\begin{equation}\label{proof_smart_MRC_2}
%\text{min}\left\{\mathbb{E}\left\{{\hat{\boldsymbol{h}}}_{k}^{H}\mathbf{C}_{\boldsymbol{q}_d}\hat{\boldsymbol{h}}_{k}\right\}\right\}=\sum_{m=1}^{M-N}{\mathbb{E}\{\mathcal{E}_{(m)}\}}.	
%\end{equation}
Hence, $\mathcal{E}_{(m)}$ is the $m$th order 
statistic, and assuming that the $\mathcal{E}_{(m)}$ are statistically independent and identically distributed, we have \cite{David} 
\begin{multline}\label{integral}
\mathbb{E}\{\mathcal{E}_{(m)}\}=\\
M\binom {M-1}{m-1}\int_{-\infty}^{\infty}{x\left[F(x)\right]^{m-1}\left[1-F(x)\right]^{M-m}dF(x)},
\end{multline}
where $x$ is the realization of $\mathcal{E}_{(m)}$ and $F(x)$ is the cumulative distribution function of $\mathcal{E}_{m}$. For the case that we have considered, where the channel coefficients are i.i.d. Rayleigh distributed, the $\mathcal{E}_{m}$ are independent Gamma random variables with 
\begin{equation}\label{gamma cdf}
	F(x)={\gamma\left(\frac{x}{\sigma_{\hat{h}}^{2}},K\right)},
\end{equation}
where $\gamma(.,.)$ denotes the incomplete Gamma function. From \cite{Nadarajah}, the integral (\ref{integral}) can be calculated in closed form for Gamma random variables as\color{black}  
\begin{equation}\label{chi_gamma}
	\mathbb{E}\{\mathcal{E}_{(m)}\}=
	{\sigma_{\hat{h}}^{2}}\chi_m.
\end{equation} 
This is in contrast to the unordered case where $\mathbb{E}\{\mathcal{E}_{m}\}=K\sigma_{\hat{h}}^{2}$.
 As a result
 \begin{equation}\label{proof_smart_MRC_3}
 	\text{min}\left\{\mathbb{E}\left\{{\hat{\boldsymbol{h}}}_{k}^{H}\mathbf{C}_{\boldsymbol{q}_d}\hat{\boldsymbol{h}}_{k}\right\}\right\}=\left(1-\frac{2}{\pi}\right)\frac{\sigma_{\hat{h}}^{2}}{K}\sum_{m=1}^{M-N}{\chi_m}.
 \end{equation}
The remaining terms in (\ref{uplink achievable rate}) can be calculated similar to the case where the high-resolution ADCs are connected to arbitrary antennas. Plugging these terms and (\ref{proof_smart_MRC_3}) into (\ref{uplink achievable rate}) and some algebraic manipulation results in (\ref{MRC_smart}).
\color{black}

\end{document}